\documentclass{amsart}
\numberwithin{equation}{section}

\usepackage[utf8]{inputenc}
\usepackage{amssymb}
\usepackage{mathtools}
\usepackage{tikz}
\usepackage{graphicx}
\usepackage{microtype}
\usepackage[margin=1in, marginparsep=0.1in, marginparwidth=0.8in]{geometry} 
\usepackage{xcolor}
\usepackage[all]{xy}
\usepackage{xspace}
\usepackage[backref=page]{hyperref}
\usepackage{hypcap}
\usepackage[capitalize, noabbrev]{cleveref}
\usepackage{tensor}
\usepackage{leftidx}

\hypersetup{
	colorlinks,
	linkcolor={red!50!black},
	citecolor={green!50!black},
	urlcolor={blue!50!black}
}

\everyentry={\displaystyle}

\newcommand{\C}{\mathbb C}

\newcommand{\Q}{\mathbb Q}
\newcommand{\R}{\mathbb R}
\newcommand{\Z}{\mathbb Z}
\renewcommand{\H}{\mathbb H}
\newcommand{\RP}{\mathbb{RP}}

\renewcommand{\O}{\mathrm O}
\newcommand{\bbS}{\mathbb{S}}

\newtheorem{thm}[equation]{Theorem}
\newtheorem{lem}[equation]{Lemma}
\newtheorem{cor}[equation]{Corollary}
\newtheorem{prop}[equation]{Proposition}

\crefname{thm}{Theorem}{Theorems}
\crefname{lem}{Lemma}{Lemmas}
\crefname{cor}{Corollary}{Corollaries}

\theoremstyle{definition}

\newtheorem{exm}[equation]{Example}
\newtheorem{defn}[equation]{Definition}

\newtheorem{ansatz}[equation]{Ansatz}

\theoremstyle{remark}
\newtheorem{rem}[equation]{Remark}

\newcommand{\vp}{\varphi}
\newcommand{\e}{\varepsilon}
\newcommand{\inj}{\hookrightarrow}
\newcommand{\surj}{\twoheadrightarrow}
\newcommand{\id}{\mathrm{id}}
\newcommand{\pt}{\mathrm{pt}}

\makeatletter

\newcommand{\@dpfr}[3][]{\frac{\partial^{#1} #2}{\partial #3{}^{#1}}}
\newcommand{\@spfr}[3][]{\partial^{#1} #2 / \partial #3{}^{#1}}
\newcommand{\pfr}[3][]{\mathchoice{\@dpfr[#1]{#2}{#3}}{\@dpfr[#1]{#2}{#3}}
{\@spfr[#1]{#2}{#3}}{\@spfr[#1]{#2}{#3}}}

\makeatother

\DeclarePairedDelimiter\paren{(}{)}
\DeclarePairedDelimiter\ang{\langle}{\rangle}
\DeclarePairedDelimiter\abs{\lvert}{\rvert}

\DeclarePairedDelimiter\bkt{[}{]}
\DeclarePairedDelimiter\set{\{}{\}}
\makeatletter
\let\oldparen\paren
\def\paren{\@ifstar{\oldparen}{\oldparen*}}
\let\oldbkt\bkt
\def\bkt{\@ifstar{\oldbkt}{\oldbkt*}}
\makeatother

\newcommand{\Bord}{\mathsf{Bord}}

\newcommand{\Vect}{\mathsf{Vect}}
\newcommand{\sAlg}{\mathsf{sAlg}}

\newcommand{\term}{\textit}
\DeclareMathOperator{\Aut}{Aut}
\DeclareMathOperator{\Hom}{Hom}

\newcommand{\GL}{\mathrm{GL}}

\newcommand{\cH}{\mathcal H}
\DeclareMathOperator{\Spec}{Spec}

\DeclareMathOperator{\End}{End}

\newcommand{\MO}{\mathit{MO}}
\newcommand{\SO}{\mathrm{SO}}

\DeclareMathOperator{\Map}{Map}

\newcommand{\fo}{\mathfrak o}

\newcommand{\MTH}{\mathit{MTH}}
\newcommand{\MTO}{\mathit{MTO}}

\newcommand{\cC}{\mathcal{C}}

\newcommand{\Sb}{S_b}
\newcommand{\Snb}{S_{\mathit{nb}}}
\newcommand{\Pin}{\mathrm{Pin}}
\newcommand{\ko}{\mathit{ko}}
\newcommand{\KO}{\mathit{KO}}
\newcommand{\MSpin}{\mathit{MSpin}}
\newcommand{\MPin}{\mathit{MPin}}
\newcommand{\Cl}{C\ell}
\newcommand{\Spin}{\mathrm{Spin}}
\newcommand{\AB}{\mathit{AB}}
\newcommand{\loc}{\mathrm{loc}}
\newcommand{\MTPin}{\mathit{MTPin}}
\newcommand{\MTSpin}{\mathit{MTSpin}}

\newcommand{\hS}{h\mathcal{S}}
\newcommand{\hSp}{h\mathcal{S}p}
\newcommand{\Top}{\mathcal{T}op}

\newcommand{\Alg}{\mathsf{Alg}}

\newcommand{\ls}[2]{\leftidx{^{#1}}{#2}{}}

\newcommand{\pinm}{pin\textsuperscript{$-$}\xspace}
\newcommand{\pinp}{pin\textsuperscript{$+$}\xspace}
\newcommand{\pinpm}{pin\textsuperscript{$\pm$}\xspace}
\newcommand{\spin}{spin\xspace}

\newcommand{\sVect}{\mathsf{sVect}}
\newcommand{\MH}{\mathit{MH}}
\newcommand{\ClHam}{\Cl(M_0', \fo)}

\begin{document}

\title{The Arf-Brown TQFT of pin\textsuperscript{$-$} surfaces}
\author{Arun Debray}
\author{Sam Gunningham}
\date{\today}

\begin{abstract}
The Arf-Brown invariant $\AB(\Sigma)$ is an 8th root of unity associated to a surface $\Sigma$ equipped with a $\Pin^-$ structure. In this note we investigate a certain fully extended, invertible, topological quantum field theory (TQFT) whose partition function is the Arf-Brown invariant. Our motivation comes from the recent work of Freed-Hopkins on the classification of topological phases, of which the Arf-Brown TQFT provides a nice example of the general theory; physically, it can be thought of as the low energy effective theory of the Majorana chain, or as the anomaly theory of a free fermion in 1 dimension. 
\end{abstract}
\maketitle

\tableofcontents 

\section{Introduction}
As part of a general program to classify and understand topological phases of matter within condensed-matter
physics, there a large body of recent work focusing on the special case of symmetry-protected topological (SPT)
phases. This classification question has been studied by many authors in different settings and with many different
approaches: for lists of references, see~\cite[\S1]{GJF17} and~\cite[\S9.3]{FH16}. It is believed that the
low-energy physics of SPT phases is often described by invertible topological quantum field theories (TQFTs), which
admit a purely mathematical classification, and that the classification of a given class of SPTs often agrees with
the classification of the analogous class of invertible TQFTs. At the same time, work on the mathematical theory of
invertible TQFTs has understood their classification as a problem in stable homotopy theory~\cite{GMTW, FH16,
SP17}. Freed-Hopkins~\cite{FH16} use this to answer the classification problem across a wide range of dimensions
and symmetry types.

In this paper, we explain this perspective on classifying invertible TQFTs and SPT phases in a specific setting,
focusing on 2-dimensional theories formulated on manifolds with a \pinm structure. Freed-Hopkins show that the
group of deformation classes of 2D invertible \pinm{} TQFTs is isomorphic to $\Z/8$, and is generated by a TQFT
$Z_\AB$ whose partition function is the Arf-Brown invariant of a \pinm surface, a generalization of the Arf
invariant of a spin surface.

In \S\ref{arf_brown_invariant}, we provide three definitions for the Arf-Brown invariant, and compare each to an
analogous definition of the Arf invariant: in \S\ref{intersection_theoretic}, the original intersection-theoretic
description due to Brown~\cite{Bro71}; in \S\ref{index_theoretic}, an index-theoretic description due to
Zhang~\cite{Zhang1, Zhang2}; and in \S\ref{KO_theoretic}, a new description using a twisted Atiyah-Bott-Shapiro
map.

Then, in \S\ref{invertible_TQFTs}, we discuss how the classification of 2D invertible TQFTs reduces to a homotopy-theoretic
problem. The fact that we're in dimension 2 allows for an explicit description of the 2-categories and homotopy
2-types that enter this argument, which are more complicated in higher dimensions. Moreover, some aspects of the
story, such as the choice of a target category and the stable homotopy hypothesis, are understood in dimension 2
but not in higher dimensions. In \S\ref{2D_generalities}, we discuss some generalities of invertible 2D TQFTs, and
in \S\ref{hom_hyp} discuss the stable homotopy hypothesis in dimension 2. We use this in \S\ref{classifying_2D} to
classify 2D invertible TQFTs of a given symmetry type.

In \S\ref{arf_brown_TQFT}, we apply this to \pinm theories: the homotopy-theoretic approach to invertible TQFTs and
the $\KO$-theoretic description of the Arf-Brown invariant combine to define the Arf-Brown TQFT $Z_\AB$. We discuss
what this theory assigns to closed \pinm $0$-, $1$-, and $2$-manifolds and how it relates to the twisted
Atiyah-Bott-Shapiro orientation.

Finally, in \S\ref{TR_Majorana}, we discuss a conjectural appearance of the Arf-Brown TQFT in physics, as the
low-energy theory of the Majorana chain. In \S\ref{physics_overview}, we provide some background on SPTs and the low-energy
approach to their classification. We then formulate the Majorana chain on an arbitrary compact \pinm 1-manifold in
\S\ref{Majdefn}, and discuss its low-energy TQFT and how it relates to the Arf-Brown TQFT in \S\ref{Majlow}. We
find that the space of ground states of the Majorana chain depends on a \pinm structure, which is expected, but
doesn't appear to have been determined before.

We also provide some preliminaries on Clifford algebras and pin manifolds in \S\ref{pin_structures}, and on the
stable homotopy theory that we use in \S\ref{homotopy_theory}.

\subsection*{Acknowledgments} 
We thank Dan Freed for many helpful conversations, as well as the organizers of the NSF-CBMS conference, David
Ayala and Ryan Grady.

\section{Preliminaries}

\subsection{Clifford algebras, pin groups, and pin structures}
\label{pin_structures}
Pin structures are generalizations of spin structures to unoriented vector bundles and manifolds. In this section,
we define the pin groups and state a few useful results about them. For proofs and a more detailed exposition,
see~\cite{ABS}.

\begin{defn}
\label{Cliffalg}
Let $k$ be a field of characteristic not equal to 2, $S$ be a finite set, and $\fo\colon S\to\set{\pm 1}$ be a
function. The \term{Clifford algebra} $\Cl(k, S,\fo)$ is defined to be the $k$-algebra
\begin{equation}
\label{cliffdefn}
	\Cl(k, S, \fo) \coloneqq T(k[S])/(s^2 = \fo(s), st = -ts\mid s,t\in S, s\ne t),
\end{equation}
where $T(k[S])$ denotes the tensor algebra of the space of functions $S\to k$, and we identify $s$ with the
function equal to $1$ at $s$ and $0$ elsewhere.

For $S \coloneqq \set{1,\dotsc,m}\cup\set{-1,\dotsc,-n}$ and $\fo(x) \coloneqq \operatorname{sign}(x)$, we'll write
$\Cl_{m,n}(k)\coloneqq\Cl(k,S,\fo)$, as well as $\Cl_n(k)\coloneqq\Cl_{n,0}(k)$ and
$\Cl_{-n}(k)\coloneqq\Cl_{0,n}(k)$. If $k = \C$, we'll suppress $\C$ from the notation, e.g.\ writing $\Cl_{m,n}$,
$\Cl_n$, and $\Cl_{-n}$.
\end{defn}
The ideal in the quotient in~\eqref{cliffdefn} contains only even-degree elements of the tensor algebra, so the
Clifford algebras are $\Z/2$-graded algebras, or \term{superalgebras}. If $a$ is a homogeneous element in a
$\Z/2$-graded algebra or module, we will let $\abs a\in\Z/2$ denote its degree.
\begin{lem}[{\cite[Proposition 1.6]{ABS}}]
Let $S_1$ and $S_2$ be finite sets and $\fo_i\colon S_i\to\set{\pm 1}$ be functions. If $\fo\colon S_1\amalg
S_2\to\set{\pm 1}$ is $\fo_i$ on $S_i$, then there is a canonical isomorphism
\[\Cl(k, S_1, \fo_1)\otimes_k\Cl(k, S_2, \fo_2)\cong \Cl(k, S_1\amalg S_2, \fo).\]
\end{lem}
For this to be true, we must use the graded tensor product, whose multiplication contains a sign: if $a,b,a',b'$
are homogeneous elements, then
\begin{equation}
	(a\otimes b)\cdot (a'\otimes b') = (-1)^{\abs{b'}\abs{a'}}aa'\otimes bb'.
\end{equation}
Let $\alpha\in\End(\Cl(k, S, \fo))$ be the \term{grading operator}, whose action on a homogeneous element $a$ is
multiplication by $(-1)^{\abs a}$.
\begin{defn}
The \term{Clifford group} is
\[\Gamma(k, S, \fo)\coloneqq\set{x\in\Cl(k, S, \fo)^\times\mid \alpha(x) y x^{-1}\in
k[S]\subset\Cl(k, S, \fo)\text{ for all } y\in k[S]}.\]
\end{defn}
Here we use the canonical map $k[S]\inj T(k[S])\surj\Cl(k,S,\fo)$, which is injective.
\begin{defn}
There is an involution $\beta\colon \Cl(k, S, \fo)\to\Cl(k, S, \fo)$ induced from the map $\widetilde \beta\colon
T(k[S])\to T(k[S])$ sending a homogeneous element
\[f_1\otimes\dotsb\otimes f_n\mapsto f_n\otimes\dotsb\otimes f_1.\]
The \term{Clifford norm} $N\colon\Gamma(k, S, \fo)\to k^\times$ is defined by $N(x)\coloneqq\beta(x)\cdot x$.
\end{defn}
\begin{defn}
The \term{pin group} $\Pin(k, S, \fo)$ associated to the Clifford algebra in \cref{Cliffalg} is the kernel of the
Clifford norm. The \term{spin group} $\Spin(k, S, \fo)$ is the subgroup of $\Pin(k, S, \fo)$ which is even in the
grading on the Clifford algebra.
\end{defn}
We are interested in the case where $k = \R$, so that the pin and spin groups are Lie groups. If we specialize to
$\Cl_{\pm n}(\R)$, they're compact Lie groups;
\begin{defn}
Let $\Pin_n^+$ denote the pin group associated to $\Cl_n(\R)$ and $\Pin_n^-$ denote the pin group associated to
$\Cl_{-n}(\R)$. The corresponding spin groups are canonically isomorphic, so we denote either one by $\Spin_n$.
\end{defn}
\begin{prop}
Let $\Pin_n^\pm$ denote either of $\Pin_n^+$ or $\Pin_n^-$. Then, there are group extensions
\begin{subequations}
\begin{align}
	\xymatrix@1{
		1\ar[r] & \Spin_n\ar[r] & \Pin_n^\pm\ar[r]^-{\pi_0} & \Z/2\ar[r] & 1
	}\\
	\label{to_On}
	\xymatrix@1{
		1\ar[r] & \Z/2\ar[r] & \Pin_n^\pm\ar[r]^\rho & \O_n\ar[r] & 1.
	}
\end{align}
\end{subequations}
\end{prop}
Let $\rho\colon H\to G$ be a homomorphism of Lie groups and $\pi\colon P\to M$ be a principal $G$-bundle. Recall
that a \term{reduction of the structure group} of $P$ to $H$ is data $(\pi'\colon Q\to M,\theta)$ such that
\begin{itemize}
	\item $\pi'\colon Q\to M$ is a principal $H$-bundle, and
	\item $\theta\colon Q\times_H G \to P$ is an isomorphism of principal $G$-bundles, where $H$ acts on $G$
	through $\rho$.
\end{itemize}
An equivalence of reductions $(Q_1,\theta_1)\to(Q_2,\theta_2)$ is a map $\psi\colon Q_1\to Q_2$ intertwining
$\theta_1$ and $\theta_2$.
\begin{defn}
If $\rho\colon H\to\GL_n(\R)$ is a Lie group homomorphism, an \term{$H$-structure} on a vector bundle $E\to X$ is
an equivalence class of reductions of the structure group of the principal $\GL_n(\R)$-bundle of frames of $E$ to
$H$. If $M$ is a smooth manifold and $E = TM$, this is called a \term{tangential $H$-structure} on $M$; if $M$ is a
smooth manifold and $E$ is its stable normal bundle, this is called a \term{normal $H$-structure}.
\end{defn}
For example, an $\SO_n$-structure is the same thing as an orientation. A \term{spin structure} on an $n$-manifold
$M$ is a tangential $\Spin_n$ structure, and we define \pinp and \pinm{} structures analogously.
\begin{rem} 
We note that such structures are \emph{stable} in the following sense: a (s)\pinpm-structure on a vector bundle $V$
is equivalent to a (s)\pinpm-structure on $V\oplus \underline{\R}$. In particular, a stable framing on a vector
bundle (a trivialization of $V \oplus \underline{\R}^N$ for some $N$) determines a (s)\pinpm-structure.
\end{rem}

\begin{prop}[{\cite[Lemma 1.3]{KT90}}]
\label{pinSW}
Let $E\to X$ be a vector bundle and $w_n(E)\in H^n(X;\Z/2)$ denote its $n^{\mathrm{th}}$ Stiefel-Whitney class.
\begin{itemize}
	\item $E$ admits a spin structure iff $w_1(E) = 0$ and $w_2(E) = 0$.
	\item $E$ admits a \pinp structure iff $w_2(E) = 0$.
	\item $E$ admits a \pinm structure iff $w_2(E) + w_1(E)^2 = 0$.
\end{itemize}
In all cases, if $E$ admits one of these structures, the set of such structures (in the spin case, with fixed orientation) on $E$ is an
$H^1(X;\Z/2)$-torsor.
\end{prop}
\begin{cor}
Let $M$ be a closed manifold of dimension at most $2$. Then, $M$ has a \pinm-structure, and has a spin structure
if and only if it is orientable. If $\dim M = 2$, then $M$ has a \pinp structure iff its Euler characteristic is even.
\end{cor}
\begin{rem}
There are a few facts about pin structures which might be surprising to a reader who has only studied spin
manifolds. A tangential spin structure is equivalent data to a normal spin structure, but this is false for pin
structures: a tangential \pinp structure is equivalent to a normal \pinm structure, and vice versa. This is
discussed in~\cite[\S1]{KT90}, and will be relevant in our homotopical approach to 2D \pinm TQFTs.

The product of spin manifolds has an induced spin structure, but using \cref{pinSW} one can write down \pinm
manifolds whose product doesn't have a \pinm structure, and similiarly for \pinp{}. This means that the \pinp and
\pinm cobordism groups are not rings, though they are modules over the spin cobordism ring.
\end{rem}
\begin{prop}[{\cite[\S2]{KT90}}]\label{prop-pin-bordism-groups}
Let $\Omega_n^{H}$ denote the cobordism group of $n$-manifolds with tangential $H$-structure. Then:
\begin{enumerate}
	\item There are isomorphisms $\Omega_1^{\Spin}\cong\Z/2$ and $\Omega_1^{\Pin^-}\cong\Z/2$, and the forgetful
	map $\Omega_1^{\Spin}\to\Omega_1^{\Pin^-}$ is an isomorphism. Both are generated by the circle with structure
	induced by its Lie group framing.
	\item There are isomorphisms $\Omega_2^{\Spin}\cong\Z/2$ and $\Omega_2^{\Pin^-}\cong\Z/8$ which identify the
	forgetful map $\Omega_2^{\Spin}\to\Omega_2^{\Pin^-}$ with the map sending $1\mapsto 4$. The torus with spin
	structure induced from its Lie group framing generates $\Omega_2^{\Spin}$, and $\RP^2$ (with either of its two \pinm structures) generates
	$\Omega_2^{\Pin^-}$.
\end{enumerate}
\end{prop}
In particular, the two isomorphism classes of \pinm{} circles aren't cobordant: one bounds and the other doesn't.
We denote the bounding \pinm circle by $\Sb^1$, and the nonbounding \pinm circle by $\Snb^1$. This applies
\textit{mutatis mutandis} to the two spin circles.

\subsection{Homotopy theory}

\label{homotopy_theory}
We follow the conventions in \cite{Beaudry-Campbell}. If the reader is unfamiliar with spectra and the stable homotopy category we recommend they first read Section 2 of loc cit. Here, we briefly recall some notation, basic definitions, and examples.

\subsubsection{The (stable) homotopy category}
\begin{itemize}
	\item The unstable homotopy category is denoted $\hS$. This category receives a map from the category $\Top$ of topological spaces and continuous maps which takes weak equivalences in $\Top$ to isomorphisms in $\hS$; by abuse of notation we will denote the image of a topological space in $\hS$ by the same symbol, and refer to the objects of $\hS$ as spaces.
	\item We write $[X,Y]$ for the space of morphisms in $\hS$ between spaces $X$ and $Y$; if we choose ``nice enough'' representatives for $X$ and $Y$ (for example CW-complexes), then this is given by the set of homotopy classes of maps between $X$ and $Y$.
	\item We denote by $\hSp$ the stable homotopy category (also known as the homotopy category of spectra). This category receives a functor from the category of prespectra which takes weak equivalences of prespectra to isomorphisms in $\hSp$. As in the unstable case, we will refer to objects of $\hSp$ simply as spectra, and use the same symbol for a prespectrum and its corresponding object in $\hSp$.
	\item Given a pair of spectra $E,F$, we write $[E,F]$ for the set of morphisms $\Hom_{\hSp}(E,F)$, and
	$[E,F]_n$ for $[\Sigma^n E, F]$; if $E$ and $F$ are ``nice enough'' (for example, if they are \term{CW-spectra}
	-- prespectra such that each space is a CW complex, the structure maps are cellular inclusions, and the
	adjoints of the structure maps are homeomorphisms) then $[E,F]$ is given by homotopy classes of maps between spectra. There is a natural abelian group structure on $[E,F]$ and $[E,F]_n$.
	\item The category $\hSp$ carries a symmetric monoidal structure $\wedge$, called the smash product.  There is also a mapping object (internal hom) $\Map(E,F)$ whose homotopy groups are $\pi_n(\Map(E,F))=[E,F]_n$.
\end{itemize}

\subsubsection{Examples of spectra}

\begin{exm}
	Given a pointed space $X$, we have the \term{suspension spectrum} $\Sigma^\infty X$, which may be presented by a prespectrum whose $n$th space is $\Sigma^nX$. A special case of this construction is the \term{sphere spectrum} $\bbS = \Sigma^\infty S^0$, which is the unit object for the smash product. 
\end{exm}

Given a spectrum $E$ and a space $X$, we write $E^n(X)$ for the \term{$E$-cohomology group} $[\Sigma^\infty X, E]_n$. Similarly, we write $E_n(X)$ for the \term{$E$-homology group} $\pi_n(\Sigma^\infty X\wedge E)$. These are examples of \term{generalized cohomology} (resp.\ \term{homology}) \term{theories}: they satisfy all of the Eilenberg-Steenrod axioms except the dimension axiom.
\begin{rem}
The Brown representability theorem~\cite{Brown75} states that any generalized cohomology theory $h^*$ (resp. generalized homology theory $h_*$) arises from a spectrum $E$ in this manner; we say that $h^*$ (resp.\ $h_*$) is \term{represented by} $E$.
\end{rem}
\begin{exm}
Ordinary cohomology with coefficients in an abelian group $A$ is represented by the \term{Eilenberg-MacLane spectrum} $HA$. This may be modeled as a spectrum whose $n$th space is the Eilenberg-MacLane space $K(n,A)$ for $n > 0$, and whose nonpositive spaces are trivial.
	
Complex $K$-theory is a generalized cohomology theory represented by a spectrum denoted $\mathit{KU}$. Similarly, real $K$-theory is represented by a spectrum $\mathit{KO}$.
\end{exm}

A spectrum is called \emph{connective} if it has trivial negative homotopy groups. Given a connective spectrum $E$, its zeroth space\footnote{Here it is essential that one considers a spectrum rather than just a prespectrum (i.e.\ the adjoints of the structure maps are homeomorphisms).} has the structure of an infinite loop space. In fact, the homotopy theory of connective spectra is equivalent to that of infinite loop spaces: given an infinite loop space, by definition it has a sequence of deloopings which form the spaces in the corresponding spectrum. See Adams~\cite{InfiniteLoopSpaces} for more on this correspondence.

\subsubsection{Thom spaces and Thom spectra}
Given a space $X$ and a vector bundle $V\to X$, one may form the \term{Thom spectrum} $\tau(V)$, which is the suspension spectrum of the Thom space. Thom spectra satisfy the property
\[
\tau(V \oplus W) \simeq \tau(V) \wedge \tau(W).
\]

We extend this definition to include \emph{virtual} vector bundles (formal differences of vector bundles). For example, there is a Thom spectrum $\tau(-V)$. Explicitly, one can find another bundle $W$ such that $V\oplus W \simeq \underline\R^N$ for some $N$; then, $\tau(-V) \simeq \Sigma^{-N}\tau(W)$.

The classifying space $BO_n$ of the $n$th orthogonal group carries a universal vector bundle $\gamma_n\coloneqq
EO_n\times_{\O_n}\R^n\to BO_n$. We denote by $\MO_n$ the Thom spectrum $\tau(\gamma_n)$; taking the colimit,
we have a spectrum $\MO$, represented by a prespectrum whose $n$th space is the Thom space of $\gamma_n$. One
defines analogous spectra $\MH_n$ and $\MH$ for any family of groups $H_n$ with compatible homomorphisms $H_n \to \O_n$ (such as $\SO_n$, $\Spin_n$, and $\Pin^\pm_n$).

The \emph{Madsen-Tillmann} spectrum $\MTO_n$ is defined as the Thom spectrum of the virtual bundle $-\gamma_n$.
There are maps $\Sigma^n MTO_n \to \Sigma^{n+1}\MTO_{n+1}$, and the direct limit is denoted $MTO$. One defines
$\MTH_n$ and $\MTH$ analogously for groups $H_n$ as above.

The reason we care about Thom spectra is that their homotopy groups compute cobordism.
\begin{thm}[Thom~\cite{ThomThesis}, Pontrjagin~\cite{Pon55}]
\label{PTthm}
Let $\rho_n\colon H_n\to\O_n$ be a compatible family of Lie group homomorphisms. Then there are isomorphisms
\begin{itemize}
	\item $\pi_n(\MTH)\cong\Omega_n^H$, the cobordism group of $n$-manifolds with tangential $H_n$-structure, and
	\item $\pi_n(\MH)\cong\Omega_n^{\nu H}$, the cobordism group of $n$-manifolds with normal $H_n$-structure.
\end{itemize}
\end{thm}
In the case when tangential $H_n$-structure is the same thing as normal $H_n'$-structure, the relevant Thom spectra
are weakly equivalent. In particular, there are weak equivalences $\MO\simeq\MTO$, $\MSpin\simeq\MTSpin$,
$\MPin^+\simeq\MTPin^-$, and $\MPin^-\simeq\MTPin^+$.

Thom spectra are also useful for understanding duality in $\hSp$.
\begin{thm}[Atiyah~\cite{AtiyahThomComplexes}]
Let $M$ be a closed manifold. Then $\Sigma^\infty M$ is dualizable in $\hSp$, and its dual is the Thom spectrum of
the stable normal bundle $\nu$ of $M$.
\end{thm}
We won't say very much about duality, but we note in particular that if $B$ is a spectrum with dual $B^\vee$, then
for any spectra $A$ and $C$ there is a weak equivalence
\begin{equation}
	\Map(A \wedge B, C)\simeq \Map(A, B^\vee\wedge C).
\end{equation}
For more on duality, see~\cite[\S III.5]{AdamsStableHomotopy}.

\section{The Arf-Brown invariant of a \pinm surface}
In this section, we give various constructions of the Arf-Brown invariant of a \pinm surface: intersection
theoretic in \S\ref{intersection_theoretic}, index-theoretic in \S\ref{index_theoretic}, and $KO$-theoretic in
\S\ref{KO_theoretic}.
 
\label{arf_brown_invariant}

	\subsection{Intersection-theoretic descriptions of the invariants}
		\label{intersection_theoretic}
		The Arf invariant of a spin surface and the Arf-Brown invariant of a \pinm surface are complete cobordism
invariants defined using intersection theory. 

\subsubsection{The Arf invariant of a spin surface}
Let $\Sigma$ be a closed oriented surface. If $x,y\in H_1(\Sigma;\Z/2)$, then the mod 2 intersection number
$I_2(x,y)\in\Z/2$ is defined by choosing smooth, transverse representative curves for $x$ and $y$ and computing the
number of points mod 2 in their intersection. This does not depend on the choice of representatives and defines a
non-degenerate bilinear pairing 
\[
I_2\colon H_1(\Sigma;\Z/2)\otimes H_1(\Sigma;\Z/2)\to\Z/2.
\]
A \term{$\Z/2\Z$-quadratic enhancement} of $I_2$ is a quadratic form on $H_1(\Sigma;\Z/2)$ whose induced bilinear form is $I_2$. Explicitly, this is a function
\[
q\colon H_1(\Sigma;\Z/2)\to\Z/2
\]
such that for all $x,y\in H_1(\Sigma;\Z/2)$,
\begin{equation}
	q(x+y) = q(x) + q(y) + I_2(x,y).
\end{equation}
The set of $\Z/2\Z$-quadratic enhancements of $I_2$ is an $H^1(\Sigma;\Z/2)$-torsor: given a $\gamma\in H^1(\Sigma;\Z/2)$ and a quadratic enhancement $q\colon H_1(\Sigma;\Z/2)\to\Z/2$, the function $q_\gamma(x)\coloneqq q(x) + \gamma(x)$ is again a quadratic enhancement. 

We have the following relationship between spin structures and quadratic enhancements of the intersection form.
\begin{thm}[\cite{Joh80, Ati71}]
There is an isomorphism of $H^1(\Sigma;\Z/2)$-torsors between the set of $\Z/2\Z$-quadratic enhancements of $I_2$ and isomorphism classes of \spin structure on $\Sigma$.
\end{thm}
\begin{rem}
Given a spin structure on $\Sigma$, the associated quadratic form is easy to describe: it takes a homology class represented by an embedded circle to either $0$ or $1$ depending on whether the induced spin structure on the circle is bounding or non-bounding.
\end{rem}

\begin{defn}
Given a spin surface $\Sigma$ with corresponding quadratic form $q$, the Arf invariant of $q$ may be defined as follows. If $\set{e_i,f_i}$ is a basis of $H_1(\Sigma;\Z/2)$ which is symplectic with
respect to the intersection form, then
\begin{equation}
	\operatorname{Arf}(\Sigma)\coloneqq \sum_i q(e_i)q(f_i)\in\Z/2.
\end{equation}
\end{defn}
\begin{thm}[\cite{KT90}]
The Arf invariant is a spin bordism invariant, and defines an isomorphism
\begin{equation}
	\operatorname{Arf}:\Omega_2^\Spin\to\Z/2.
\end{equation}
\end{thm}

\begin{exm}
Let $T=S^1\times S^1$ denote the torus with spin structure afforded by the Lie group framing. Consider the symplectic basis $\set{e,f}$ for $H_1(T;\Z/2)$ corresponding to the embedded circles $S^1\times \set{1}$ and $\set{1}\times S^1$. As each circle carries the non-bounding spin structure, the associated quadratic form $q$ takes the values $q(e)=q(f)=1$. Thus the Arf invariant is $1\in\Z/2$, and hence $T$ is a generator for the \spin bordism group.
\end{exm}

\subsubsection{The Arf-Brown invariant of a \pinm surface}
Now suppose $\Sigma$ is any closed surface (not necessarily oriented). Then $H_1(\Sigma; \Z/2)$ still carries a non-degenerate intersection form $I_2$, although $H_1(\Sigma; \Z/2)$ may be odd dimensional, and will not admit a symplectic basis in general. In this case, one must consider the following notion:
\begin{defn}
A \term{$\Z/4$-quadratic enhancement of the intersection form} on $\Sigma$ is a function
\begin{equation}
	q\colon H_1(\Sigma;\Z/2)\to\Z/4
\end{equation}
such that for all $x,y\in H_1(\Sigma;\Z/2)$,
\begin{equation}
	q(x+y) = q(x) + q(y) + 2\cdot I_2(x,y),
\end{equation}
where $(2\cdot)\colon\Z/2\inj\Z/4$ is inclusion.
\end{defn}
As with $\Z/2$-quadratic enhancements, the set of $\Z/4$ quadratic enhancements is an $H^1(\Sigma;\Z/2)$-torsor: given a $\gamma\in H^1(\Sigma;\Z/2)$ and a
quadratic enhancement $q\colon H_1(\Sigma;\Z/2)\to\Z/4$, the function $q_\gamma(x)\coloneqq q(x) + 2\cdot\gamma(x)$
is again a $\Z/4$-quadratic enhancement.
\begin{thm}[\cite{KT90}]
For any closed surface $\Sigma$, there is an isomorphism of $H^1(\Sigma;\Z/2)$-torsors from the set of
\pinm structures on $\Sigma$ to the set of $\Z/4$-quadratic enhancements of the intersection form on $\Sigma$.
\end{thm}

\begin{defn}[\cite{Bro71, KT90}]
Let $\Sigma$ be a \pinm surface and let $q\colon H_1(\Sigma;\Z/2)\to\Z/4$ be its associated quadratic
enhancement. The \term{Arf-Brown invariant} of $\Sigma$ is the unit complex number
\begin{equation}
	\AB(\Sigma)\coloneqq \frac{1}{\sqrt{\abs{H_1(\Sigma;\Z/2)}}}\sum_{x\in H_1(\Sigma;\Z/2)} \exp\paren{\frac{2\pi
	i q(x)}{4}}.
\end{equation}
\end{defn}
This is sometimes called the \term{Kervaire invariant} or the \term{Arf-Brown-Kervaire invariant}.
\begin{thm}[\cite{Bro71, KT90}]
The Arf-Brown invariant $\AB(\Sigma)$ is a \pinm bordism invariant, and defines an isomorphism
\begin{equation}
	\AB:\Omega_2^{\Pin^-}\to \mu_8 \cong \Z/8
\end{equation}
where $\mu_8$ denotes the group of eighth roots of unity.
\end{thm}

\begin{exm}
Let us compute the value of the Arf-Brown invariant for the \pinm structures on $\RP^2$. In that case $H_1(\RP^2;\Z/2) \simeq \Z/2$, and there are two quadratic enhancements of the intersection form which take the image of the non-zero homology class to either $1$ or $3$ $\mod 4$. In the first case, we see that the Arf-Brown invariant is $\exp(\frac{2\pi i}{8})$, and in the second $\exp(\frac{-2\pi i}{8})$. It follows that either structure gives a generator for \pinm bordism.
\end{exm}

If $\Sigma$ is an oriented surface, then a $\Z/4$-quadratic enhancement is necessarily valued in the even elements of $\Z/4$, and thus recovers the $\Z/2$-quadratic enhancement corresponding to a spin structure. Moreover, the Arf-Brown invariant of such a quadratic enhancement is the (exponentiated) Arf-invariant of the corresponding quadratic form.

	\subsection{Index-theoretic description of the invariants}
	\label{index_theoretic}
	The Arf(-Brown) invariant of a \spin (or \pinm) surface admits an alternative description in terms of Dirac operators acting on sections of (s)pinor bundles, which we will now describe. In the spin case, the Arf invariant corresponds to the mod 2 index, or Atiyah invariant of a Spin Riemann surface - the mod 2 dimension of the space of holomorphic sections of a theta-characteristic. In the \pinm case, the Arf-Brown invariant may be interpreted as the reduced $\eta$-invariant of a twisted Dirac operator as defined and studied by Zhang \cite{Zhang1, Zhang2}.

\subsubsection{The Atiyah invariant of a spin surface}
Let $\Sigma$ be a closed surface equipped with a Riemannian metric and a spin structure.\footnote{Here, we consider a Riemannian metric to induce a \emph{negative} definite quadratic form on the fibers of $T^\ast_\Sigma$.} Then $\Sigma$ carries a graded spinor bundle
\begin{equation}\label{graded spinor}
S_\Sigma = P_{\Spin_2} \times_{\Spin_2} \Cl_{-2}(\R)
\end{equation}
with a left action of the bundle of Clifford algebras $\Cl(T^\ast_\Sigma)$ and a commuting right action of the constant algebra $\Cl_{-2}(\R)$. This bundle splits as a sum of its graded components $S^0_\Sigma \oplus S^1_\Sigma$, where each $S^i_\Sigma$ carries a fiberwise action of $\Cl_{-2}^0(\R) \cong \C$ and thus may be considered as a complex line bundle. 

There is a Dirac operator
\[
	D^+_\Sigma:C^\infty(S^0_\Sigma) \to C^\infty(S^1_\Sigma)
\]
given by composing the canonical connection operator
\[
\nabla: C^\infty(S_\Sigma) \to C^{\infty}(T^\ast _\Sigma \otimes S_\Sigma)
\]
with the action of sections of $T^\ast_\Sigma$ via Clifford multiplication.

\begin{defn}
	The \emph{Atiyah invariant} of the spin surface $\Sigma$ is 
	\[
	\dim \ker (D^+_\Sigma) \bmod 2 \in \Z/2.
	\]
\end{defn}

\begin{rem}	
	The Riemannian metric and orientation on $\Sigma$ determine a complex structure, and the even spinor bundle $S^0_{\Sigma}$ defines a square root of the holomorphic cotangent bundle (such a square root is known as a \emph{theta characteristic}). The Dirac operator is identified with the $\overline{\partial}$ operator defining the holomorphic structure on $S^0_\Sigma$. Thus the Atiyah invariant of $\Sigma$ is the mod 2 dimension of the space of holomorphic sections of $S^+_\Sigma$.
\end{rem} 
	
	\begin{prop}[\cite{Joh80}]
		Given a closed spin surface $\Sigma$, the Atiyah invariant and Arf invariant coincide.
	\end{prop}

\subsubsection{Reduced $\eta$-invariant of a \pinm surface}
 Given a \pinm surface $\Sigma$, the pinor bundle 
 \[
 S_\Sigma = P_{\Pin^-_2} \times_{\Pin^-_2} \Cl_{-2}(\R)
 \]
 still makes sense, though it doesn't carry a natural $\Z/2$-grading. However, as Clifford multiplication is not \pinm equivariant, the formula for the Dirac operator now defines a map on sections
 \[
D_\Sigma: C^\infty(S_\Sigma) \to C^\infty(S_\Sigma \otimes \delta),
 \]
 where $\delta$ is the orientation bundle. 
 
To get an operator acting on sections of the same bundle, we may apply the following trick (which is spelled out in \cite{Stolz_exotic} in the analogous case of a \pinp-manifold of dimension 4 (mod 8)). The left regular action of $\Cl_{-2}(\R) \cong \H$ on itself extends to an action of $\Cl_{-3}(\R) \cong \H \oplus \H$; now compose the operator $D_\Sigma$ with the action of $e_3 \in \Cl_{-3}(\R)$ to get a self-adjoint operator $\widetilde{D}_\Sigma$ on sections $S_\Sigma$ called the \emph{twisted Dirac operator}.
 		
	The twisted Dirac operator has an associated $\eta$-function, defined for $s\in \C$ with $Re(s) \gg 0$ by the formula
	\[
	\eta_{\widetilde{D}}(s) = \sum_{\lambda \neq 0} \mathrm{sign}(\lambda) \dim E(\lambda) |\lambda|^{-s}
	\]
	where $E(\lambda)$ is the eigenspace with eigenvalue $\lambda$.
	This function admits a meromorphic extension to $s=0$, and thus we may define the \emph{reduced $\eta$-invariant}:
		\[
		\overline{\eta}(\widetilde{D}) = \frac{\dim \ker(\widetilde{D}) + \eta_{\widetilde{D}}(0)
		}{2} \bmod 2\Z \in \R/2\Z
	\]	
	
	\begin{prop}[\cite{Zhang1, Zhang2}]
	Given a \pinm surface $\Sigma$, the reduced $\eta$-invariant $\overline{\eta}(\widetilde{D}_\Sigma)$ is an element of $\Z[\frac 14]/2\Z$, and agrees with the Arf-Brown invariant of $\Sigma$ under the isomorphism given by the exponential map $\Z[\frac{1}{4}]/2\Z \cong \mu_8 \subseteq \C^\times$.
	\end{prop}

\begin{rem}
	In the \spin case, the contribution from the $\eta$-invariant vanishes, and we are just left with half the dimension of $\ker(D_\Sigma)$ (or equivalently the dimension of $\ker(D^+_\Sigma)$) mod 2.
\end{rem}

	\subsection{$\KO$-theoretic descriptions of the invariants}
		\label{KO_theoretic}





Here we explain how the analytic index-theoretic invariants of the previous section may be expressed topologically in terms of pushforwards in (twisted) $\KO$-theory.

\subsubsection{The Atiyah-Bott-Shapiro orientation and pushforward maps in $\KO$-theory}\label{sec:the-atiyah-bott-shapiro-orientation-and-pushforward-maps-in-ko-theory}


Let $\pi:V\to X$ be a rank-$k$ real vector bundle equipped with a \spin structure, and let $\tau(V)$ denote its Thom space. The Clifford module construction of Atiyah-Bott-Shapiro~\cite{ABS} determines a
Thom isomorphism
\begin{equation}
\label{spinThom}
\KO^*(X)\xrightarrow{\cong} \widetilde\KO^{*+k}(\tau(V)).
\end{equation}
This isomorphism is given by multiplication by a Thom class $u_V \in \widetilde{\KO}^n(\tau(V))$, which may be described as follows. The \spin structure on $V$ determines a graded spinor bundle $S_V = S_V^0 \oplus S_V^1$ (see (\ref{graded spinor})), which carries a left action of the Clifford bundle $\Cl(V)$. Pulling $S_V$ back to the total space of $V$, we obtain a pair of bundles together with a homomorphism
 \[
 \pi^\ast(S^1_V) \to \pi^\ast(S^0_V)
 \]
given by Clifford multiplication, which is an isomorphism away from the zero section. This defines an element of $\KO(V,V-\{0\}) \cong \widetilde{\KO}(\tau(V))$ which is the required Thom class.

Using the Thom isomorphism, we can define a pushforward for an $n$-dimensional spin manifold $M$. Choose an embedding $M\to\R^N$
for some large $N$, and let $\nu\to M$ be the normal bundle, which has rank $N-n$. Using the tubular
neighborhood theorem to embed $\nu\inj\R^N$, consider the Pontryagin-Thom collapse map 
\[
PT_\nu: S^{N} = \R^N \cup \{\infty\} \to \tau(\nu)
\] 
which takes the complement of the tubular neighborhood in $S^N$ to the basepoint of the Thom space. 

%
	
\begin{defn}	
The \term{pushforward map in $\KO$-theory} for $X$ is the composition
	\begin{equation}
	\label{pushforward_KO}
	\xymatrix{
		\pi_!^M\colon \KO^*(M)\ar[r]_-{\eqref{spinThom}} &
		\widetilde\KO^{*+N-n}(\tau(\nu))\ar[r]^-{PT_\nu^\ast}
		&\widetilde{\KO}^{*+N-n}(S^N)\ar[r]^-s_-{} &\KO^{*-n}(\pt).
	}
	\end{equation}
\end{defn}
where $s$ is the suspension isomorphism.
One may check that this invariant does not depend on the choice of embedding for large enough $N$. Moreover, by considering the appropriate modification for manifolds with boundary, one can check that the association of $\pi^M_!(1_M) \in KO^{-n}(pt)$ to a closed \spin $n$-manifold is a cobordism invariant.
\begin{rem}
	Another perspective on the pushforward in $KO$-theory is that a closed Spin $n$-manifold $M$ carries a
	fundamental class $[M] \in \KO_n(M)$, Spanier-Whitehead dual to the Thom class in $u_M \in
	\widetilde{\KO}^{-n}(\Sigma^{-N}\tau(\nu))$. We can then pushforward to get a class in $\KO_n(pt) = \KO^{-n}(pt)$.
\end{rem}

The collection of Thom classes given by the Atiyah-Bott-Shapiro orientation may be succinctly formulated as a morphism of spectra (in fact of $E_\infty$-ring spectra ~\cite{Joa04, AHR})
\begin{equation}
	\widehat A\colon\MSpin\to\KO,
\end{equation}
The Pontryagin-Thom construction can then be interpreted as a homomorphism (in fact, isomorphism) from the spin cobordism groups $\Omega_n^\Spin$ to the homotopy groups $\pi_n(\MSpin)$, as in \cref{PTthm}; the induced map on homotopy groups
\[
\Omega_n^{\Spin}\cong \pi_n(MSpin) \to \pi_n(\KO) = \KO^{-n}(pt)
\]
takes the class of a closed spin $n$-manifold $M$ to the pushforward $\pi^M_! 1_M$.



With the pushforward in hand, we can give a description for the Arf invariant of a spin surface which is simpler, if less intuitive,
than the one given in the previous section.
\begin{prop}[\cite{Ati71}]
Let $\Sigma$ be a spin surface. The Atiyah/Arf invariant of $\Sigma$ is the pushforward of $1$:
\begin{equation}
	A'(\Sigma)\coloneqq \pi_!^\Sigma(1)\in\KO^{-2}(\pt)\cong\Z/2.
\end{equation}
\end{prop}

\begin{rem}
From this perspective, the Arf invariant is an example of a generalized characteristic number, specifically a
\term{$\KO$-Pontrjagin number} as constructed by Anderson-Brown-Peterson~\cite{ABP66}.
\end{rem}





\subsubsection{Twisted Thom isomorphism for \pinm manifolds}\label{twisted Thom}
Next we discuss the generalization of this invariant to \pinm surfaces. One issue is that \pinm vector bundles are not oriented for $KO$, so it is not immediately clear how to define a pushforward.

The key idea is the observation that a \pinpm structure on a vector bundle $V$ is equivalent to a \spin structure on the virtual vector bundle $V \mp \det(V)$ (see \cite{KT90}). In particular, we have a Thom class
\[
u_{V\mp \det(V)} \in \KO^n(\tau(V \mp \det(V))
\]
Now if $M$ is a closed \pinm $n$-manifold with an embedding in $\R^N$ for $N\gg 0$, its normal bundle $\nu$ is equipped with a \pinp structure, so we have a corresponding Thom class
\[
u_M \in \KO^{N-n-1}(\tau(\nu - \delta)) = \KO^{N-n}(\tau(\nu + 1 - \delta)
\]
Alternatively, by Spanier-Whitehead duality, there is a fundamental class
\[
[M] \in \KO_n(\tau(\delta-1)).
\]

These ideas are best understood from the perspective of twisted $KO$-theory. 
\begin{defn}
	Given a space $X$ equipped with a map $w:X\to BO_1$ (which we may think of as classifying either a double cover $X^w$ or a real line bundle $L_w$), we define the \emph{twisted $\KO$-cohomology}
	\[
	\KO^{w + n}(X) = \KO^n(\tau(L_w-1)).
	\]
	Similarly, we have the \emph{twisted $KO$-homology}
	\[
	\KO_{w+n}(X) = \KO_{w+n}(\tau(1-L_w)).
	\]
\end{defn}

\begin{rem}
This construction is an example of the notion of a twisted generalized cohomology theory. This and other examples can be found in \cite{FHT1} (e.g.\ see Example 2.28.)	and~\cite{ABG10}.
\end{rem}

	\subsubsection{Pushforward in twisted $\KO$-homology}\label{pushforward twisted KO}
	Using this language, we may now give a $\KO$-theoretic construction of the Arf-Brown (or reduced $\eta$-invariant) of a \pinm surface $\Sigma$. Let
	\[
	w_1:\Sigma \to BO_1
	\]
	denote the classifying map of the orientation line bundle. Then consider the pushforward
	\begin{equation}
\label{w1push}
	w_{1!}([\Sigma]) \in \KO_{2+w}(BO_1)
	\end{equation}
	
The group 
\[
\KO_{2+w}(BO_1) = \KO_2(\Sigma^{-1}\MO_1)\cong \KO_3(\RP^\infty)
\]
is isomorphic to the direct limit of cyclic 2-groups $\Z/2^\infty$, or equivalently, the collection of all $2^n$th
roots of unity in $\C^\times$. One may also apply the above construction for the connective theory $\ko$, in which
case (see \cite{BG10})
\[
\ko_{2+w}(BO_1) \cong \Z/8
\] 
and thus the class $w_{1!}([\Sigma])$ may be interpreted as an 8th root of unity via the exponential map.
	\begin{prop}
The class $w_{1!}([M]) \in\ko_{2+w}(BO_1)\cong \Z/8$ agrees with the Arf-Brown invariant of $\Sigma$ (via the isomorphism between $\Z/8$ and the 8th roots of unity in $\C^\times$ given by the exponential map).
	\end{prop}

\begin{rem}\label{remark about defining index using $KO$-homology}
One may define an invariant in $\Z/8$ associated to any vector bundle on $\Sigma$ by first considering the twisted Poincar\'e duality isomorphism
\[
\KO^0(\Sigma) \to \KO_{2+w_1}(\Sigma)
\]
then taking the pushforward in twisted $\KO$-homology as before.
\end{rem}

\begin{rem}
The advantage of this perspective is that it naturally occurs as the induced map on homotopy groups of a map of spectra. Namely, the observation from \S\ref{twisted Thom} relating \pinm structures and \spin structures leads to the identification
\[
\MTPin^- \simeq \MPin^+ \simeq \MSpin \wedge \Sigma^{-1}\MO_1.
\]
Thus smashing the Atiyah-Bott-Shapiro orientation with the factor $\Sigma^{-1}\MO_1$ leads to a map of spectra
\[
\MTPin^- \to \ko \wedge \Sigma^{-1}\MO_1.
\]
The class $w_{1!}([M])$ as defined in \ref{w1push} is precisely the image of the class in $\Omega_2^{\Pin^-}$ defined by $M$.
\end{rem}

We now mention a few closely related $\KO$-theoretic constructions of the Arf-Brown invariant that appear in the literature.
\subsubsection{Zhang's construction}
Given a \pinm surface $\Sigma$, the classifying map of the orientation double cover  
\[
w_1: \Sigma \to BO_1 \simeq \RP^\infty
\]
is homotopic to a map which factors through the two skeleton $\RP^2$. Let $\widetilde{w_1}$ denote the corresponding map $\Sigma \to \RP^2$. The stable normal bundle of $\widetilde{w_1}$ is a virtual vector bundle of rank 0 which carries a spin structure (after choosing a \pinm structure on $\RP^2$), and thus there is a well-defined pushforward map
\[
\KO^\ast(\Sigma) \to \KO^\ast(\RP^2).
\]
In \cite{Zhang2}, Zhang defines the topological index of a vector bundle $V$ on $\Sigma$ as the image of $f_!V$ under a certain homomorphism
\[
\KO^0(\RP^2) \cong \Z \oplus \Z/4 \to \Z[\frac{1}{4}]/2 \cong \Z/8.
\]
It is shown in loc cit.\ that the topological index agrees with the reduced $\eta$-invariant of a twisted Dirac operator on $V$. In particular, the topological index of the trivial bundle agrees with the Arf-Brown invariant.

One can check that Zhang's construction of the topological index agrees with the one in Remark \ref{remark about defining index using $KO$-homology} (first observe that in each case, the index factors through $\KO^0(\RP^2)$, then check that the morphism to $\Z/8$ is the same).

\subsubsection{Distler-Freed-Moore construction}
Distler-Freed-Moore~\cite{DFM10} give a slightly different $\KO$-theoretic construction. In order to state it, we will need to consider two modifications of the cohomology theory $\KO$: first, we consider the Postnikov truncation of the connective cover $R=\KO\langle 0,1,2,3,4\rangle$; then we take $R$-cohomology with coefficients in $\R/\Z$ (this is represented by a spectrum $R(\R/\Z)$ which fits in to an exact triangle $R \to R\wedge H\R \to R(\R/\Z)$). Moreover, the authors consider a twisted version of this theory, as explained above. 

It is shown in \cite{DFM10} that $R^{w_1-2}(BO_1;\R/\Z)$ is cyclic of order 8. Given a closed \pinm surface $\Sigma$, the Arf-Brown invariant is obtained as the image of a generator under the morphisms
\[
R^{w_1-2}(BO_1;\R/\Z) \xrightarrow{w_1^\ast} R^{w_1-2}(\Sigma;\R/\Z) \xrightarrow{\pi^\Sigma_!} R^{-4}(pt;\R/\Z) \cong \R/\Z \xrightarrow{\exp} \C^\times.
\]

\section{Invertible TQFTs via stable homotopy theory}
	\label{invertible_TQFTs}

	In this section, we give a brief exposition of TQFTs leading to the classification of invertible TQFTs in terms of homotopy theory. For concreteness we will focus on the case of 2-dimensional TQFTs. The notion of fully extended TQFT requires some ideas from higher category theory, which we will discuss only at an informal level. However, when considering invertible theories, we will see that the higher categorical aspects may be reformulated in the (perhaps) more familiar terms of stable homotopy theory.

\subsection{What is a (2d) invertible TQFT?}
\label{2D_generalities}

\subsubsection{Atiyah-Segal Axioms} In \cite{atiyah1988} (see also~\cite{SegalCFT}) an $n$-dimensional TQFT was axiomatized as a symmetric monoidal functor
\[
Z:\Bord_{n,n-1} \to \Vect_\C
\] 
The source category has objects given by closed $(n-1)$-manifolds, and morphisms are cobordisms between them. Thus a TQFT $Z$ will assign a vector space $Z(Y^{n-1})$ to a closed $(n-1)$-manifold, and a linear map 
\[
Z(X):Z(Y_1)\to Z(Y_2)
\]
whenever $X$ is a cobordism between $Y_1$ and $Y_2$. Moreover, this assignment is compatible with the symmetric monoidal structures on each side: disjoint union of manifolds is taken to tensor product of vector spaces. In particular, the empty $(n-1)$-manifold $\emptyset^{n-1}$ is the unit object for the symmetric monoidal structure on $\Bord_{n,n-1}$, and thus we can identify $Z(\emptyset^{n-1})$ with the trivial line $\C$. Given a closed $n$-manifold, $X^n$, we may interpret it as a cobordism between empty $(n-1)$-manifolds, and thus we obtain a number $Z(X) \in \C$; this is referred to as the \term{partition function} of the theory.

In order to study the Arf-Brown TQFT, we will need a number of variations, extensions, and simplifications of these axioms.

\subsubsection{Fully extended TQFTs}  A 2-dimensional TQFT as defined above may be thought of as an assignment of an invariant $Z(X)\in \C$ to every closed $2$-manifold $X$, which satisfies a certain locality property. Namely, $Z(X)$ may be computed by decomposing $X$ in to pairs of pants and discs. We will be interested in \emph{fully extended TQFTs} which satisfy a stronger form of locality, allowing the partition function to be computed by decomposing $X$ in to arbitrarily small pieces (for example, a triangulation).

One way to make this precise involves replacing the ordinary category $\Bord_{2,1}$ with a certain 2-category $\Bord_2$. In this paper, we will use the term 2-category to denote what some authors would call a \emph{weak $2$-category}, or \emph{bicategory}. We will treat the subject in an informal and expository manner---a 2-category should have objects, $1$-morphisms, and $2$-morphisms which may be composed in various ways; for further details, see ~\cite{benabou_introduction_1967, SP17}. In our case, the objects of $\Bord_2$ are now $0$-manifolds, $1$-morphisms are $1$-dimensional cobordisms, and $2$-morphisms are diffeomorphism classes of certain $2$-manifolds with corners, interpreted as cobordisms between $1$-dimensional cobordisms. 

Similarly we must replace the category $\Vect_\C$ with an appropriate 2-category of coefficients for the theory. One choice is the Morita $2$-category $\Alg_\C$, whose objects are algebras, $1$-morphisms are bimodules, and $2$-morphisms are bimodule homomorphisms. As in the usual Atiyah-Segal axioms, these $2$-categories carry symmetric monoidal structures (for more details on symmetric monoidal 2-categories, see~\cite{Shu10}).  A fully extended 2d TQFT may be formulated as a symmetric monoidal functor:
\[
Z:\Bord_2 \to \Alg_\C
\]
Note that the original categories $\Bord_{1,2}$ and $\Vect_\C$ sit as the endomorphism categories of the unit objects in $\Bord_2$ and $\Alg_\C$ respectively, and thus a fully extended TQFT gives rise a TQFT as in Atiyah and Segal's original definition.

\subsubsection{Tangential structure} As the Arf-Brown invariant is an invariant of \pinm surfaces, to define the Arf-Brown TQFT, we must consider a variant of the bordism category in which all the manifolds are equipped with a \pinm structure. More generally, for any Lie group $H$ equipped with a homomorphism $H\to O_2$, there is a 2-category $\Bord_2^{H}$ defined as above, but now all manifolds are equipped with ${H}$-structures.\footnote{One defines a ${H}$-structure on a $0$ or $1$ manifold by taking the direct sum of the tangent bundle with a trivial bundle of appropriate rank (so the total rank is two). Or better, one can consider the manifolds appearing in the bordism category as being equipped with 2-dimensional collars. Making this precise in the smooth category takes some care; see \cite{SP17} for a careful treatment in the 2-dimensional case.}  

\begin{rem}[Structure groups]\label{rem-symmetry-groups}
In the cases of interest, the group $H$ is usually part of a family $H_n$, $n \in \Z_{>0}$, equipped with maps $H_n \to O_n$. For example, we could take $H_n$ to be $SO_n$, $\Spin_n$, $\Pin_n^{\pm}$, or $O_n$ itself. These examples all have the additional property that an $H_n$-structure on a vector bundle $V$ is equivalent to an $H_{n+1}$-structure on $V\oplus \underline\R$; thus we may think of an $H_n$ structure as an $H$-structure on the corresponding stable vector bundle (where $H = \varinjlim H_n$). In that case, we will generally write $\Bord_2^{H}$ instead of $\Bord_2^{H_2}$. We will also consider the framed case $H_n = 1$; however, note that a framing on an $n$-dimensional bundle is not determined by a stable framing in general.
\end{rem}
\begin{rem}\label{remark duals}
	Given a symmetric monoidal 2-category $\cC$, we say that $\cC$ \emph{has duals} if every object of $\cC$ admits a dual with respect to the monoidal structure and every 1-morphism of $\cC$ admits an adjoint. The bordism 2-categories $\Bord_2^{H_2}$ all have duals in this sense. Moreover, one version of the cobordism hypothesis \cite{} states that the framed bordism category $\Bord_2^{fr}$ (obtained by taking ${H_2}=1$ above) is the universal with this property, in the following sense: Given any symmetric monoidal 2-category with duals $\cC$, symmetric monoidal functors $\Bord_2^{fr} \to \cC$ are in correspondence with objects of $\cC$ (where the correspondence assigns to a functor, the value of the framed $0$-manifold $pt_+$).
\end{rem}

\subsubsection{Gradings}
To allow for more interesting theories we can also enlarge the coefficient category by considering algebras and bimodules with a $\Z/2\Z$-grading (and morphisms compatible with this grading). This gives rise to a 2-category $\sAlg_\C$ (the $2$-category of superalgebras), which is equipped with a symmetric monoidal structure incorporating the Koszul rule of signs. Later we will see that this choice of target category is universal in a certain sense  (see Corollary \ref{cor-salg-universal}).

\subsubsection{Examples}
\begin{exm}[Euler Theories]\label{example-euler}
	There is a TQFT $Z_1$ defined on unoriented manifolds which assigns only identity objects and morphisms in $s\Alg$. In particular, the partition function takes the constant value $1$. Slightly more interesting is the \emph{Euler theory} $Z_\lambda$, associated to a complex number $\lambda \in \C^\times$. This agrees with the trivial theory on $0$ and $1$-manifolds, but assigns the number $\lambda^{\chi(X)}$ (considered as a linear map between trivial lines) to any 2-dimensional cobordism. 
\end{exm}

\begin{exm}[2-dimensional Dijkgraaf-Witten Theory \cite{dijkgraaf_topological_1990, freed_topological_2010}]
	Given a finite group $G$, there is a 2d oriented TQFT whose partition function on a closed surface is the weighted sum
	\[
	Z_G(\Sigma) = \sum_{[P]} \frac{1}{|\Aut(P)|} 
	\]
	This theory assigns the group algebra of $G$ to a point, and the space of class functions to a circle.
\end{exm}
 
\begin{exm}[The Arf-Brown Theory]\label{exm-arf-brown}
We will see in \ref{arf_brown_theory} that the Arf-Brown invariant is the partition function of a 2d fully extended TQFT. In other words, there is a symmetric monoidal functor
\[
Z_{AB}:\Bord_2^{Pin^-} \to \sAlg
\]
such that for a \pinm surface $X$, $Z_{AB}(X)=\AB(X)$. The Arf-Brown theory takes the following values on lower dimensional manifolds:
\begin{itemize}
\item For a bounding \pinm circle $S^1_{b}$, $Z_{AB}(S^1_b)\cong \C$, an even line.
\item For a non-bounding \pinm circle $S^1_{nb}$, $Z_{AB}(S^1_{nb})\cong \C[1]$, an odd line.
\item We have $Z_{AB}(pt)\cong \Cl_1$, the first Clifford algebra.
\end{itemize}
\end{exm}

%


\subsubsection{Invertible theories} The Euler theories and the Arf-Brown theory have the property that every value of the functor $Z$ is invertible (either as an object with respect to the monoidal structure, or as a $1$ or $2$-morphism) in the symmetric monoidal 2-category $\sAlg$ (for example, Dijkgraaf-Witten theory does not have this property unless $G=1$). Such TQFTs are called \emph{invertible}.

When dealing with invertible theories we may restrict attention to the following class of 2-categories:
\begin{defn}
	A $2$-category is called a \emph{$2$-groupoid} if every $1$-morphism and $2$-morphism is invertible. A symmetric monoidal 2-category $\cC$ is called a \emph{Picard 2-groupoid} if it is a 2-groupoid and in addition every object is invertible with respect to the monoidal structure.  
\end{defn}
Note that if a symmetric monoidal 2-groupoid $\cC$ has duals, then it is necessarily a Picard 2-groupoid, i.e. the duals of objects must be inverses.

Given a symmetric monoidal $2$-category $\cC$, we may consider the maximal subcategory $\cC^\times$ which is a Picard 2-groupoid (i.e. throw out any non-invertible objects and morphisms). A TQFT 
\[
Z: \Bord_2^{H_2} \to s\Alg_\C
\]
 is invertible if and only if $Z$ factors through $s\Alg_\C^\times \to s\Alg_\C$. 

There is another way to associate a 2-groupoid  $\ls{\loc}{\cC}$ to a $2$-category $\cC$ by formally inverting any non-invertible $1$ and $2$-morphisms. This procedure is left adjoint to the inclusion of $2$-groupoids in to $2$-categories (the maximal $2$-groupoid is right adjoint). In other words, any functor from $\cC$ to a 2-groupoid will factor uniquely through $\ls\loc\cC$. Moreover, if $\cC$ is symmetric monoidal and has duals, then $\ls\loc\cC$ will be a Picard 2-groupoid. We note the following upshot: the data of an invertible (fully extended) $H$-TQFT is given by a functor of Picard 2-groupoids
\[
Z: \ls\loc\Bord_2^H \to \sAlg^\times
\]

\subsection{The homotopy hypothesis and stable 2-types}
\label{hom_hyp}
Now let us explain how an invertible TQFT may be reformulated in terms of maps in the stable homotopy category. 

\subsubsection{The fundamental $2$-groupoid and Postnikov truncations}
To begin with let us recall the following $2$-category associated to a space.
\begin{exm}
	Let $X$ be a topological space. The \emph{fundamental 2-groupoid} $\pi_{\leq 2}(X)$ of $X$ is given as follows:
	\begin{itemize}
		\item The objects are points of $X$,
		\item The $1$-morphisms are paths between points,
		\item The $2$-morphisms are given by homotopy classes of homotopies between paths.
	\end{itemize}
\end{exm}

To understand what information about a space the fundamental $2$-groupoid captures, let us recall the following:
\begin{defn}
We say that a topological space $X$ is a homotopy $n$-type if the homotopy groups $\pi_i(X)$ are non-zero only if $i= 0, 1,\ldots n$. Given any space $X$, there is a Postnikov fibration
\[
f_{\leq n}: X \to X_{\leq n}
\]
where $X_{\leq n}$ is an $n$-type and $f_{\leq n}$ induces an isomorphism on $\pi_i$ for $i=0,1, \ldots n$. We refer to $X_{\leq n}$ as the (Postnikov) $n$-truncation of $X$, or simply the $n$-type of $X$.
\end{defn}

\subsubsection{Unstable homotopy hypothesis}
The idea of the homotopy hypothesis (formulated by Grothendieck in ``Pursuing Stacks'') is that the homotopy theory of $n$-types should be modeled by $n$-groupoids. There are many forms the homotopy hypothesis might take, depending on what flavor of higher category theory one considers. In some cases, the result is almost tautological (if ones uses an inherently homotopical theory of higher category), and in others it is false (if one uses a too strict a higher category theory). 

In the case $n=2$, one expects\footnote{We were unable to locate a suitable reference for this particular statement.}  that the following categories are equivalent:
\begin{itemize}
	\item The homotopy category of $2$-types;
	\item The category of $2$-groupoids, with equivalence classes of $2$-functors.
\end{itemize}
The equivalence is constructed as follows: to a homotopy $2$-type $X$, we assign the fundamental $2$-groupoid. The inverse functor is given by the \emph{classifying space} of a $2$-groupoid. In general, the classifying space of a $2$-category assigns a space $|\cC|$ to a $2$-category $\cC$ (see \cite{Duskin}), constructed as the geometric realization of a certain simplicial set---the \emph{nerve} of $\cC$.

\begin{rem}
	The classifying space of a 2-category $\cC$ is weakly equivalent to the classifying space of its localization $\ls{\loc}{\cC}$ (in fact, one may construct the localization by taking the fundamental 2-groupoid of its classifying space). 
\end{rem}

\subsubsection{Stable homotopy hypothesis}
A symmetric monoidal structure on a $2$-category $\cC$ induces a binary operation on the classifying space $|\cC|$. This operation is commutative and associative up to homotopy; more precisely, it can be shown that $|\cC|$ carries an $E_\infty$ structure. If $\cC$ has duals (for example, if it is a Picard $2$-groupoid), then every object has an inverse up to homotopy, and the $E_\infty$ structure is said to be grouplike. 

Foundational results in homotopy theory (see~\cite[\S2.3]{InfiniteLoopSpaces} and the references therein, or~\cite[\S5.1.3]{HigherAlgebra} for a modern approach) state that a grouplike $E_\infty$-structure on a space $X$ is equivalent to an infinite loop space structure on $X$. Equivalently, $X = \Omega^\infty E$ may be identified as the zeroth space in a connective spectrum $E$ (the other spaces in the spectrum are given by iterated deloopings or loop spaces of $X$).

\begin{defn}
	A \emph{stable $n$-type} is a connective spectrum $E$ such that $\Omega^\infty E$ is a homotopy $n$-type.
\end{defn}

Thus we arrive at the \emph{stable homotopy hypothesis}, which states that there is an equivalence between:
\begin{itemize}
	\item The homotopy category stable $2$-types.
	\item The category of Picard $2$-groupoids with equivalence classes of symmetric monoidal functors.
\end{itemize}
Thanks to recent work of \cite{GJO17} this is now a precisely formulated theorem.

\subsection{Classifying invertible TQFTs up to isomorphism}
\label{classifying_2D}
One consequence of the homotopy hypothesis is that we may encode a 2d $G$-TQFT as a morphism of stable $2$-types:
\[
|Z|:|\Bord_2^G| \to |\sAlg_\C^\times|
\]
Thus to classify isomorphism classes of invertible TQFTs, we should classify homotopy classes of maps between stable $2$-types as above. 

\subsubsection{Stable Postnikov data}	
Let us first unpack the case of a stable $1$-type, which according to the appropriate version of the the homotopy hypothesis, is equivalent data to a Picard groupoid. 

A stable $1$-type $E$ may be succinctly encoded in terms of the following data:
	\begin{itemize}
		\item A pair of abelian groups $A=\pi_0E$ and $B=\pi_1E$,
		\item A homomorphism $k:A/2A \to B$
	\end{itemize}
The homomorphism $k$ encodes the $k$-invariant
\[
HA \to \Sigma^2 HB
\]
which defines the Postnikov tower of $E$. 
\begin{rem}[{\cite[\S3]{GJOS}}]\label{rem-action-of -eta}
The homomorphism $k$ may also be identified with the action of the generator $\eta \in \pi_1(\bbS) \cong Z/2$ on the homotopy groups of $E$ (see Remark \ref{rem-ring-spectra}).
\end{rem}

\begin{prop}[\cite{JO12}]\label{prop-picard}
Given a Picard $1$-groupoid $\cC$, we recover the above data as follows:
\begin{itemize}
	\item $\pi_0(|\cC|)$ is the set of equivalence classes of objects in $\cC$, with group structure coming from the symmetric monoidal structure.
	\item $\pi_1(|\cC|)$ is the set of automorphisms of the unit object in $\cC$.
	\item Given an object $a\in \cC$, the element $k(a) \in \pi_1(|\cC|)$ may be identified with the image of the symmetry map $\sigma \in \Aut(a\otimes a)$ under the equivalence $\Aut(1_\cC) \simeq \Aut(a \otimes a)$ induced by the functor $- \otimes (a\otimes a)$.
\end{itemize}
\end{prop}

Similarly, a stable $2$-type may be described by the following data:
\[
\xymatrix{
	\Sigma^2 HC \ar[r]^{i_2} & E = E\langle 0,1,2\rangle \ar[d] & \\
	\Sigma^1 HB \ar[r]^{i_1} & E\langle 0,1\rangle \ar[d] \ar[r]^{k_1} & \Sigma^3 HC \\
	\Sigma HA \ar[r]^{i_0} & E\langle 0\rangle \ar[r]^{k_0} & \Sigma^2 HB \\
	}
\]
As we have seen considering stable $1$-types, the maps $k_0 i_0$ and $k_1 i_1$ are classified by homomorphisms
\begin{equation}\label{eqk0}
A/2A \to B
\end{equation}
and 
\begin{equation}\label{eqk1}
B/2B \to C
\end{equation}

The map $i_0$ is an equivalence, so $k_0$ is determined by $k_0 i_0$. However, these data do not fix the homotopy class of $k_1$ itself and thus do not fix the homotopy type of $E$ in general.\footnote{For example, consider the Postnikov truncation $ku_{\leq 2}$ of connective complex $K$-theory. We have that $\pi_1(ku)=0$, so both \ref{eqk0} and \ref{eqk1} are necessarily zero; however, there is a non-zero $k$-invariant $H\Z \to \Sigma^3 H\Z$.}

\subsubsection{Bordism and Madsen-Tillmann Spectra}
Given a Lie group $H_n$ with a map $H_n\to O_n$, Galatius-Madsen-Tillmann-Weiss \cite{GMTW} define a spectrum $MTH_n$ as the Thom spectrum of the virtual vector bundle $-\gamma_n$ over $BH_n$ (where $\gamma_n$ is pulled back from the tautological vector bundle on $BO_n$). Building on the fundamental work of \cite{GMTW,MW} \cite{SP17} shows that the classifying space of $\Bord_2^{H_2}$ is given by the Postnikov 2-truncation of $\Omega^{\infty -2}MTH_2$.\footnote{More generally, in \cite{SP17} it is shown that the classifying space of the $(\infty, n)$-category $\Bord_n^{H_n}$ is $\Sigma^n MTH_n$.}

Let us also consider the following variant of the Bordism 2-category. Now suppose $H_n$ is a family of groups as in Remark \ref{rem-symmetry-groups}. Associated to such data we have the \emph{stable bordism $2$-category}, a symmetric monoidal $2$-category 
$\Bord_{2,st}^H$ 
defined in a similar way to $\Bord_2^{H_2}$, except the manifolds are now equipped with a $H$-structure on their stabilized tangent bundles (the direct limit of the sequence obtained by taking iterated direct sums with the trivial vector bundle), and crucially that 2-morphisms are given by 2-cobordisms modulo the relation given by 3-cobordisms (as opposed to diffeomorphism as before). 

Remarkably, the stable bordism $2$-category is already a Picard $2$-groupoid (the duals in the usual cobordism 2-category are inverses modulo cobordism). Its classifying space is the stable 2-type associated to the spectrum $MTH$, the direct limit of spectra $\Sigma^n MTH_n$, which represents the cohomology theory defined by $G$-cobordism of manifolds:
$\pi_i(MTH) = \Omega_i^H$ consists of smooth closed $i$-manifolds, with a $H$-structure on the stable tangent bundle, up to $H$-cobordism.

\begin{exm}
	In the case $H=1$, we have the \emph{framed} stable 2d-bordism category,
	$\Bord_{2,st}^{fr}$. Pontryagin-Thom theory identifies the classifying space of this category with the stable $2$-type of the sphere spectrum $\bbS$. The homotopy groups are given by
	\[
	\pi_0(\bbS) \cong \Z, \pi_1(\bbS) \cong \Z/2, \pi_2(\bbS) \cong \Z/2
	\]
	The generator $\eta$ of $\pi_1(\bbS)$ is represented by the stably framed manifold $S^1$ with its Lie group framing, and the generator $\eta^2$ of $\pi_2(\bbS)$ is represented by $S^1\times S^1$, also with its Lie group framing (see \ref{rem-ring-spectra})
\end{exm}

Let $H_n$ be a family of groups as in Remark \ref{rem-symmetry-groups}. There is a canonical functor 
\[
\Bord_2^{H_2} \to \Bord_{2,st}^H
\]
\begin{defn}
A $H_2$-TQFT
\[
Z: \Bord_2^{H_2}\to \sAlg_\C
\]
is called \emph{stable} if it factors through $\Bord^H_{2,st}$.
\end{defn}

\begin{exm}\label{pin-bordism-category-example}
We have described the homotopy groups of 
\[
|\Bord_{2,st}^{\Pin^-}| \simeq \MTPin^-\langle 0,1,2\rangle
\]
in Proposition \ref{prop-pin-bordism-groups}; in particular, recall that $\pi_2(\MTPin^-) = \Omega_2^{\Pin^-} \simeq \Z/8$. By the results of \cite{SP17,GMTW}, there is an equivalence:
\[
|\Bord_2^{\Pin^-}| \simeq \Sigma^2 MT\Pin_2^-\langle 0,1,2\rangle
\]
According to \cite{RW}, we have $\pi_2(\Sigma^2 MT\Pin_2^-) \cong \Z \oplus \Z/4$, and the map
\[
\Bord_2^{\Pin^-} \to \Bord_{2,st}^{\Pin^-}
\]
induces the map $\Z \oplus \Z/4 \to \Z/8$ which takes $(a,b)$ to $a+2b \mod 8$. The projection on to the first factor is represented by the Euler characteristic. In particular, we see that the Euler theories $Z_\lambda$ of Example \ref{example-euler} is stable if and only if $\lambda = \pm 1$.
\end{exm}

\begin{rem}\label{rem-ring-spectra}
The Thom spectra $MO$, $MSpin$, and $\bbS$ carry natural $E_{\infty}$-ring structures; geometrically, the ring operations corresponds to direct product of manifolds. Moreover, any spectrum is a module spectrum for the sphere spectrum in a unique way (just as any any abelian group is a module for the integers). In particular, the graded ring $\bigoplus_i \pi_i(\bbS)$ acts on $\bigoplus_i \pi_i(E)$ for any spectrum $E$. In the case $E$ is  bordism spectrum then this action may be understood in terms of direct product of appropriately structured manifolds. Note however, that $\MPin^\pm$ are not ring spectra - only module spectra for $\MSpin$.
\end{rem}

\subsubsection{Brown-Comenetz duality and invertible superalgebras}
We define a spectrum $I\C^\times$, which is a variant of the Brown-Comenetz dual of the sphere spectrum \cite{BC} (in the original construction $\Q/\Z$ was used in place of $\C^\times$). Abstractly, the spectrum $I\C^\times$ may be defined as follows in terms of its corresponding cohomology theory: for a spectrum $E$ we have
\[
(I\C^\times)^i(E) = \pi_i(E)^\vee = \Hom(\pi_i(E),\C^\times)
\]
where we write $A^\vee$ to denote the Pontryagin dual of an abelian group $A$. By construction, we have
\[
\pi_i(I\C^\times) = (I\C^\times)^{-i}(\bbS) = \pi_{-i}(\bbS)^\vee
\]
(in particular, $I\C^\times$ has vanishing homotopy groups in positive degree).

Consider the shifted spectrum $\Sigma^k I\C^\times$. This is characterized by the following universal property among spectra $E$:
\[
[E,\Sigma^k I\C^\times] \cong \Hom(\pi_k(E),\C^\times) = \pi_k(E)^\vee
\]
In particular, a character $c_k\in \pi_k(E)^\vee$ determines morphisms for each positive integer $i$:
\[
c_{k-i}: \pi_{k-i}(E) \to \pi_{k-i}(\Sigma^k I \C^\times) = \pi_i(\bbS)^\vee
\]
Unwinding the definitions, we see that these morphisms are computed by the action of $\pi_\ast(\bbS)$ on $\pi_\ast(E)$:
\begin{lem}\label{lem-bcdual-homotopy}
Given $E$ and $c_k \in \pi_k(E)^\vee$ as above, we have
\[
c_{k-i}(x)(\xi) = c_k(\xi.x)
\]
for all $x\in \pi_{k-i}(E)$, $\xi \in \pi_i(\bbS)$.
\end{lem}

The following key result identifies the Picard 2-groupoid corresponding to $\Sigma^2I\C^\times$.
\begin{prop}\label{prop-salg-bcdual}
	The stable $2$-type $|\sAlg_\C^\times|$ is equivalent to the connective cover of $\Sigma^2 I\C^\times$.
\end{prop}
We give a proof of this result in \ref{proof} below.

As an immediate consequence of Proposition \ref{prop-salg-bcdual} we obtain the following notable result:
\begin{cor}\label{cor-salg-universal}
Any character $c:\pi_2(MTH_2) \to \C^\times$ arises as the partition function of a unique invertible TQFT
\[
Z_c:Z:\Bord_2^G \to \sAlg_\C
\]
\end{cor}

\begin{exm}
Continuing with Example \ref{pin-bordism-category-example}, we see that isomorphism classes of invertible \pinm TQFTs are given by
\[
(\Z \oplus \Z/4)^\vee = \C^\times \times \mu_4
\] 
On the other hand, the stable theories are represented by the cyclic subgroup generated by $(\zeta_8,\zeta_4)$ (where $\zeta_4$ and $\zeta_8$ denote a primitive roots of unity).
\end{exm}

\subsubsection{Deformation classes of theories and the Freed-Hopkins classification}
\label{deformation_classes}
	The natural computation from the perspective of topological phases is not to do with \emph{isomorphism classes} of theories, but rather \emph{deformation classes}. Informally this may be understood in terms of replacing the $2$-category $s\Alg_\C$ with an appropriate topological $2$-category in which $\C$ is considered with its continuous topology. More precisely, one considers the \emph{Anderson dual} spectrum $I\Z$ which may be defined as the homotopy fiber of the map $H\C \to I\C^\times$ given by the exponential map. In particular, there is a map 
	\[
	|s\Alg_\C^\times| \simeq \Sigma^2I\C^\times \to \Sigma^3 I\Z
	\]
	Which can be understood as passing from the discrete to the continuous topology on $\C$. Deformation classes of 2-dimensional invertible theories may then be computed in terms of homotopy classes of maps in to $\Sigma^3I\Z$. 
	
	In fact, the relevant computation for the classification of symmetry protected phases concerns the deformation classes of \emph{reflection positive} theories. It is shown in \cite{FH16} that the deformation class of such a theory may always be represented by a stable theory. In the case of interest in this paper (2-dimensional \pinm theories), deformation classes of reflection positive theories are in natural bijection with isomorphism classes of stable theories: both groups are cyclic of order 8.

\subsubsection{Proof of \ref{prop-salg-bcdual}}\label{proof}
By the defining property of $I\C^\times$, there is a unique map of spectra 
	\[
	c:|\sAlg_\C^\times| \to \Sigma^2I\C^\times
	\]
	which induces the identity map 
	\[
	c_2 = \pi_2(c):\C^\times = \pi_2(|\sAlg_\C^\times|)\to  \pi_2(\Sigma^2 I \C^\times) = \pi_0(\bbS)^\vee = \C^\times
	\]
	To prove the proposition, we must check that the morphism $c_2$ induces isomorphisms
	\[
	c_1: \Z/2\Z \cong\pi_1(|\sAlg_\C) \to \pi_1(\Sigma^2 I\C^\times) = \pi_1(\bbS)^\vee \cong \mu_2
	\]
	and
	\[
	c_0: \Z/2\Z \cong \pi_0(\sAlg_\C) \to \pi_0(\Sigma^2 I\C^\times) = \pi_2(\bbS)^\vee \cong \mu_2
	\]
	
	By Lemma \ref{lem-bcdual-homotopy}, to understand the maps $c_0$ and $c_1$, we must compute the action of $\pi_1(\bbS)$ and $\pi_2(\bbS)$ on $\pi_\ast(|\sAlg_\C^\times|)$.

	
	First consider the generator $\eta \in \pi_1(\bbS)\cong \Z/2\Z$. Recall from Remark \ref{rem-action-of -eta} that the action of $\eta$ on the homotopy groups of the classifying space of a Picard groupoid is given by the formula in Proposition \ref{prop-picard} (which also encodes the unique $k$-invariant of the stable $1$-type). Consider the Picard groupoid $s\Vect_\C^\times$, whose classifying space is $\Omega|\sAlg_\C^\times|$. As explained in Proposition \ref{prop-picard}, the action of $\eta$ is given by
	\[
	\pi_1(\sAlg_\C) \otimes \Z/2\Z = \pi_0(s\Vect_\C^\times) \otimes \Z/2\Z \cong \Z/2\Z \to \pi_1(s\Vect_\C^\times) = \C^\times
	\]
	which takes the generator (represented by an odd line) to the number $-1 \in \C^\times$. It follows that $c_1$
	applied to the class of on odd line gives the unique non-trivial character of $\pi_1(\bbS) \cong \Z/2\Z$, and thus $c_1$ is an isomorphism as required.
	
	Now $\pi_2(\bbS) \cong \Z/2\Z$ is generated by $\eta^2$. Thus, the action of $\eta^2$ is a composite:
	\[
	\pi_0(|\sAlg_\C^\times|) \xrightarrow{\eta} \pi_1(|\sAlg_\C^\times|) \xrightarrow{\eta} \pi_2(|\sAlg_\C^\times|)
	\]
	Thus it remains to compute the map between $\pi_0$ and $\pi_1$. For this, we consider the Picard groupoid representing the $\langle 0,1\rangle$ Postnikov truncation of $|\sAlg_C^\times|$; its objects are Morita invertible complex superalgebras, and morphisms are \emph{isomorphism classes} of invertible bimodules. The action of $\eta$ is computed via the same method as before using the symmetry isomorphism. One sees that the generator of $\pi_0(|\sAlg_\C^\times|)$ (represented by the Clifford algebra $\Cl_1$) is taken by $\eta$ to the non-trivial class $\pi_1(|\sAlg_\C^\times|)$ represented by the odd line.
	
	Putting everything together we observe that $c_0$ takes the class of $\Cl_1$ to the unique non-trivial character of $\pi_2(\bbS)$, and thus is an isomorphism as required.

%
%
%

\section{The Arf-Brown TQFT}
\label{arf_brown_TQFT}
\subsubsection{The Arf-Brown Theory}\label{arf_brown_theory}
In particular, recall the Arf-Brown invariant
\[
\AB:\Omega_2^{\Pin^-} \to \mu_8 \subseteq \C^\times
\]
We denote by
\[
Z_{AB}:\Bord_2^{\Pin^-} \to s\Alg_\C
\]
the unique TQFT with partition function given by $\AB$ (see Corollary \ref{cor-salg-universal}).

\begin{prop}
	The \pinm Arf-Brown TQFT
	\[
	Z_{AB}:\Bord_2^{\Pin^-} \to s\Alg_\C
	\]
	 assigns the following invariants:
	\begin{itemize}
		\item To a \pinm point, $Z_{AB}$ assigns the first Clifford algebra $\Cl_1$.
		\item To a bounding \pinm circle, $Z_{AB}$ assigns an even line $\C$.
		\item To a non-bounding \pinm circle, $Z_{AB}$ assigns an odd line $\C[1]$.
	\end{itemize}
\end{prop}
\begin{proof}
	We use Lemma \ref{lem-bcdual-homotopy} to compute the values of $Z_{AB}$ on closed manifolds in terms of the action of $\pi_\ast(\bbS)$ on $\pi_\ast(\MTPin^-)$. 
	
	Recall that the homotopy groups of $\bbS$ (respectively $MTPin^-$) are given by cobordism classes of stably framed manifolds (respectively, \pinm-manifolds). As usual, let $\eta$ denote the generator of $\pi_1(\bbS)$, which is represented by the cobordism class of the circle $S^1$ with its Lie group framing (which induces the non-bounding \pinm structure). Thus the operation ``multiplication by $\eta$'' on $\pi_\ast(MTPin^-)$ may be understood as direct product with the \pinm manifold $S^1_{nb}$.
	
	We have \cite{KT90} that the class of $S^1_{nb}$ is a generator of $\pi_2(MTPin^-) \cong \Z/2\Z$. The class of $S^1_{nb} \times S^1_{nb}$ is the unique element of order 2 in $\pi_2(MTPin^-)$, and its Arf-Brown invariant is $-1 \in \C^\times$.

    It follows that $Z_{AB}$ takes the non-bounding circle to the unique non-trivial character of $\pi_1(\bbS)$ (which takes the class $\eta$ represented by $S^1$ to $-1$). Similarly, $Z_{AB}$ takes the \pinm point to the unique non-trivial character of $\pi_2(\bbS)$ (which takes the class $\eta^2$ represented by $S^1\times S^1$ to $-1$), as required.
    \end{proof}

\begin{rem}
The Arf-Brown theory gives rise to an invertible \spin TQFT as the composite
\[
Z_{A}: \Bord_2^{\Spin} \to \Bord_2^{\Pin^-} \to s\Alg_\C
\]
which assigns the Arf invariant to a closed Spin surface. This TQFT was studied extensively in \cite{Gunningham}.
\end{rem}

\subsubsection{The Atiyah-Bott-Shapiro orientation revisited}

First let us recall from Section \S\ref{sec:the-atiyah-bott-shapiro-orientation-and-pushforward-maps-in-ko-theory} that the Arf invariant of a closed Spin surface $\Sigma$ may be constructed as a pushforward in $KO$-theory:
\[
\pi^\Sigma_! 1_\Sigma \in ko^{-2}(pt) \cong \Z/2\Z
\]

As explained in \S\ref{sec:the-atiyah-bott-shapiro-orientation-and-pushforward-maps-in-ko-theory} the pushforward map in $ko$-theory is constructed using the Atiyah-Bott-Shapiro orientation of $KO$, which may be encoded as a map of spectra
\[
\widehat{A}: MTSpin \to ko
\]

In fact (as shown in \cite{Gunningham}) the entire Arf theory factors naturally through the Atiyah-Bott-Shapiro orientation:\footnote{The map $Cliff$ is so called because it takes an element of $ko(pt)$, represented by a finite dimensional vector space $V$, to the (complex) Clifford algebra associated to the a positive definite inner product on $V$.}
\[
Z_A: |\Bord_2^{\Spin}| \to MTSpin \xrightarrow{\widehat{A}} ko \xrightarrow{Cliff} |s\Alg_\C^\times|
\]

As explained in \S\ref{twisted Thom}, the Arf-Brown invariant may also be interpreted as a pushforward in (twisted) $\KO$-theory. One may reinterpret this as arising from the following twisted form of the Atiyah-Bott-Shapiro orientation:
\[
\MTPin^- \simeq \MTSpin \wedge \Sigma^{-1}MO_1 \xrightarrow{\widehat{A}\wedge id} ko \wedge \Sigma^{-1} MO_1
\]
The factor $\Sigma^{-1}MO_1$ is the Thom spectrum of the virtual vector bundle over $BO_1$ corresponding to the representation sphere $S^{\sigma -1}$. 

Note that the twisted Atiyah-Bott-Shapiro orientation induces an isomorphism
\[
\pi_2(\MTPin^-) \cong \pi_2(\ko \wedge \Sigma^{-1}MO_1).
\]
It follows that the Arf-Brown theory factors through the twisted ABS map:
\[
|\Bord_2^{\Pin-}| \to \MTPin^- \to ko \wedge \Sigma^{-1}ko \to |s\Alg_\C^\times|
\]

\begin{rem}
Freed-Hopkins~\cite{FH16} define an involution on $\Sigma^2\MTPin_2^-$ obtained by taking a \pinm structure to its
opposite (obtained by tensoring with the orientation double cover), and an involution on $\abs{\sAlg_\C^\times}$
induced by complex conjugation. The homotopy fixed point spectra are
$(\Sigma^2\MTPin_2^-)^{h\Z/2}\simeq\Sigma^2\MTSpin_2$ and $\abs{\sAlg_\C^\times}^{h\Z/2}\simeq
\abs{\sAlg_\R^\times}$. A 2D invertible \pinm TQFT is said to have \term{reflection structure} if the map of
spectra $\Sigma^2\MTPin_2^-\to\abs{\sAlg_\C^\times}$ it defines is equivariant with respect to these
involutions~\cite{FH16} (see also~\cite{JF17}).

The Arf-Brown theory naturally carries a reflection structure, which follows from the arguments of
Freed-Hopkins~\cite{FH16}. This is also hinted at by physics: in \S\ref{TR_Majorana} we discuss how the $\Z/8$
classification of deformation classes of 2D invertible \pinm{} TQFTs is conjecturally linked to the $\Z/8$
classification of 2D \pinm{} SPT phases. Some physics-based classifications of these SPTs~\cite{FK11, GJF17} are
rooted in real superalgebra, tying the $\Z/8$ classification to the 8-fold periodicity of Morita equivalence
classes of real Clifford algebras.
\end{rem}


%

\section{The time-reversal-invariant Majorana chain}
\label{TR_Majorana}
The Arf-Brown TQFT is believed to arise in physics as part of the classification of topological phases of matter.
In this section, we discuss one of its conjectural appearances, as the low-energy theory of the Majorana chain with
time-reversal symmetry, and some background on this occurrence.

\subsection{Symmetry-protected topological phases}
\label{physics_overview}
Condensed-matter theorists are interested in classifying topological phases of matter: given a dimension and a
collection of symmetries to be preserved (called the \term{symmetry type}), what physical systems can occur, and
what kind of data is needed to specify one up to a suitable notion of equivalence? This is a difficult problem in
general, but can be simplified by restricting to nice subclasses of phases.

This problem is complicated by the lack of a mathematical definition of a topological phase. Nonetheless, arguments
from physics suggest some properties that a definition will have: for example, given two topological phases with
the same dimension and symmetry type, it should be possible to formulate them both on the same ambient space but
with no interactions between them, creating another phase. This commutative monoid-like operation is called
\term{stacking}.
\begin{defn}
A \term{symmetry-protected topological (SPT) phase} is a topological phase of matter which is invertible under
stacking: after stacking with some other phase, it  equivalent to the trivial phase.\footnote{We haven't provided a
definition of the trivial phase, and the definition will depend on one's model for topological phases. But it
should have no interesting physics, and its partition functions on closed manifolds should all be equal to 1.}
\end{defn}
Though this isn't a mathematical definition, it tells us that equivalence classes of SPTs should form an abelian
group. The computation of this abelian group given a dimension and symmetry type has been the subject of
considerable recent research activity at the interface of topology and physics (for a long list of references,
see~\cite[\S1]{GJF17}).
\begin{rem}
The original definition in physics of a symmetry-protected phase is one which is inequivalent to the trivial theory
when its symmetry type is considered, but which is equivalent in the absence of symmetry. According to this
definition, the trivial phase is not an SPT, so the group structure is lost. Our interest in the group structure
motivates us to allow the trivial phase.
\end{rem}
To classify SPTs, one generally needs a model for phases of matter and equivalence between them.\footnote{We note,
however, the existence of model-independent approaches~\cite{Xio17, XA17, GJF17}.} Lattice models are a common
choice: roughly speaking, an $n$-dimensional lattice model is a way of assigning to any closed $n$-manifold $M$
with a simplicial structure the following data:
\begin{itemize}
	\item a complex vector space $\cH$ determined by local combinatorial data on $M$, called the \term{state
	space}; and
	\item a self-adjoint operator $H\colon\cH\to\cH$ also determined by local combinatorial data, called the
	\term{Hamiltonian}.
\end{itemize}
A lattice model is \term{gapped} if there is an $\e > 0$ such that as the simplicial structure is refined on any
closed $n$-manifold $M$, the difference between the two smallest eigenvalues of $H$ is greater than $\e$. Two
gapped lattice models are equivalent if one can be deformed into the other through local deformations of $\cH$ and
$H$ that preserve a gap in $\Spec H$.

The symmetry type corresponds to a choice of tangential structure on $M$, expressed in terms of the simplicial
structure. Here are some examples.
\begin{itemize}
	\item The default symmetry type fixes an orientation on $M$, expressed through a consistent local orientation
	of its simplices.
	\item A phase has \term{time-reversal symmetry} if we can choose $M$ to be unoriented. In this case one doesn't
	need orientations on simplices. Alternatively, because lattice models are built from local data, one can
	formulate the model on a simplicial disc, together with an explicit action of reflection on $\cH$; this is how
	time-reversal symmetry is implemented for the Majorana chain.
	\item There is a notion of a \term{fermionic SPT} which is believed to correspond to spin structure; see
	\S\ref{combspin} below.
	\item Given a finite group $G$, an \term{internal $G$-symmetry} corresponds to the data of a principal
	$G$-bundle $P\to M$. This can be formulated as a function from the 1-simplices of $M$ to $G$ encoding the
	monodromy of $P$ or by placing a simplicial structure on $P$ itself~\cite{OMD16}.
\end{itemize}
These symmetries may interact in nontrivial ways: for example, there are two ways to implement time-reversal
symmetry in fermionic phases, corresponding to \pinp{} and \pinm{} structures on $M$.
\begin{rem}
The above is not a rigorous mathematical definition of topological phases of matter. Providing a rigorous
framework for this classification problem is a significant open problem in this field.
\end{rem}
See~\cite{Sab18} for more about lattice models.

There are many approaches to classifying SPTs. We will use a low-energy limit approach, because it reduces modulo a
conjecture to a completely mathematical problem, the classification of TQFTs.
\begin{defn}
Given a gapped lattice model with Hamiltonian $H$, its space of \term{ground states} on a closed manifold $M$ is
the eigenspace for the smallest eigenvalue of $H$.
\end{defn}
In examples, this depends on the underlying manifold but not its triangulation, behaving like a topological field
theory. Conjecturally, it \emph{is} (part of) a topological field theory:
\begin{ansatz}
\label{bigansatz}
Given a $d$-dimensional lattice model with symmetry type $H_d$, there is a fully extended, reflection positive
$(d+1)$-dimensional TQFT\footnote{In general one must also allow TQFTs tensored with an invertible, non-topological
theory; see~\cite[\S5.4]{FH16}. However, this will not come into play in this paper.} $Z$ with the same symmetry
type, called the \term{low-energy (effective) field theory} of the lattice model, whose deformation class can be
determined from the data of the lattice model, and such that
\begin{enumerate}
	\item if $N$ is a closed $d$-manifold, $Z(N)$ is isomorphic to the space of ground states of the lattice model
	on $N$;
	\item if $\vp\colon N\to N$ is a diffeomorphism and $N_\vp$ denotes its mapping torus, there is a well-defined
	action of $\vp$ on the ground states of the lattice model on $N$, and $Z(N_\vp)$ is the trace of this action.
\end{enumerate}
In addition,
\begin{enumerate}
	\setcounter{enumi}{2}
	\item deformation-equivalent lattice models should have deformation-equivalent low-energy effective field
	theories, and
	\item if $S_0$ and $S_1$ are lattice models with low-energy theories $Z_0$ and $Z_1$, respectively, the
	low-energy theory of $S_0\otimes S_1$ should be $Z_0\otimes Z_1$.
\end{enumerate}
It is believed that the map sending a lattice model to its low-energy theory is surjective onto the set of
deformation classes of fully extended, reflection positive $(d+1)$-dimensional $H_d$-TQFTs.
\end{ansatz}
For discussion of this prediction, see~\cite{FH16, RW17, Gaiotto}; for discussion of reflection positivity in the
invertible case, see~\cite[\S8.2]{FH16}. For the rest of this section, we assume
\cref{bigansatz}.

\Cref{bigansatz} implies in particular that the group of equivalence classes $d$-dimensional SPTs with a given
symmetry type is isomorphic to the group of deformation classes of reflection positive invertible
$(d+1)$-dimensional TQFTs with the same symmetry type, a fact which Freed-Hopkins~\cite[\S9.3]{FH16} use to
classify fermionic SPTs. This approach to classifying SPTs is also undertaken in~\cite{PWY17, Cam17, DT18}.

\subsubsection{Context for the Majorana chain}
We now specialize to the group of 2D \pinm{} SPTs, which is isomorphic to $\Z/8$. This can be proven assuming
\cref{bigansatz}, as in~\cite[(9.7.7)]{FH16}: we saw in \S\ref{deformation_classes} that the group of deformation
classes of 2D reflection positive invertible \pinm{} TQFTs is $[\MTPin^-, \Sigma^3 I\Z]\cong\Z/8$. Other
approaches to this $\Z/8$ classification can be found in~\cite{GW14, KTTW15, BWHV17, CSRL17, GJF17}.

The Majorana chain is a 2D fermionic SPT phase with time-reversal symmetry making it into a \pinm{} phase, and is
believed to be a generator of the $\Z/8$ of such phases. It was originally studied by Kitaev~\cite{Kit01} with an
eye towards applications in quantum computing, then given time-reversal symmetry by Fidkowski-Kitaev~\cite{FK10}
and Turner-Pollmann-Berg~\cite{TPB11}, who observed that it generated a $\Z/8$ of SPTs. Therefore, \cref{bigansatz}
implies that its low-energy field theory is a tensor product of an odd number of copies of the Arf-Brown theory. In
what follows, we will formulate the Majorana chain on a \pinm{} 1-manifold and study its low-energy behavior.
%
%
%
%
%
%
\begin{rem}
There's an additional way in which the Arf-Brown theory is expected to arise in physics. Though we won't
discuss it in detail, we'll point the interested reader to some references.

Associated to a free fermion theory in dimension $d$ is its anomaly theory, a $(d+1)$-dimensional invertible field
theory of the same symmetry type. The group of equivalence classes of 1-dimensional free fermion theories with
\pinm{} symmetry is conjecturally isomorphic to $\Z$ with the Majorana chain as a generator, and its anomaly theory
is conjecturally the Arf-Brown TQFT. For general free fermion systems, this conjecture is due to
Freed-Hopkins~\cite[\S9.2.6]{FH16}; Witten~\cite[\S5]{Wit16} provides a physical argument specifically for the
time-reversal symmetric Majorana chain.

These two appearances of the Arf-Brown TQFT from the Majorana chain are believed to be related: one can regard a
free system as an interacting system with the same dimension and symmetry type, defining a group homomorphism from
equivalence classes of free fermion theories to SPTs. For 2D \pinm{} theories, this is believed to be the quotient
map $\Z\to\Z/8$, a surprising fact first noticed by Fidkowski-Kitaev~\cite{FK10, FK11} and
Turner-Pollmann-Berg~\cite{TPB11}, and argued a different way by You-Wang-Oon-Xu~\cite{YWOX14}.
Freed-Hopkins~\cite[\S\S9.2, 9.3]{FH16} provide a conjecture describing this homomorphism in general, then study it
in several specific cases, including 2D \pinm{} theories.
\end{rem}

\subsection{Combinatorial spin and pin structures}
\label{combspin}
The Majorana chain is a fermionic SPT. This corresponds mathematically to building the state space and Hamiltonian
using superalgebra. In relativistic quantum field theory, the spin-statistics theorem implies such a system should
be formulated on spin manifolds, but in the condensed-matter setting, the theorem doesn't apply. Nonetheless, it
appears that spin structures are the correct setting for fermionic phase of matter, in that the data of a fermionic
phase of matter depends on a choice of spin structure on the underlying manifold in examples~\cite{GK16}.

Time-reversal symmetry can act on fermionic phases in two ways: by squaring to $1$ or by squaring to the grading
operator.  The former is believed to correspond to a \pinm structure on the underlying manifold, and the latter to
a \pinp structure~\cite{KTTW15}.

The Majorana chain admits a time-reversal symmetry $T$ squaring to $1$, so to formulate the Majorana chain on a
compact 1-manifold $M$ with this symmetry, we must choose a \pinm structure on $M$ and encode it in the data of the
lattice somehow. In general, this is somewhat tricky: for spin structures, this was solved by
Cimasoni-Reshetikhin~\cite{CR07} in dimension 2 and Budney~\cite{Bud13} in all dimensions, but the analogue for
\pinm structures appears to be unknown. Since we're only studying 1-manifolds, we can use an explicit, simpler
construction: there are two \pinm structures on a closed interval relative to a fixed \pinm structure on its
boundary, so we will fix \pinm structures on the vertices of $M$ and use the data of the class of the \pinm
structure on the edge.

Fix, once and for all, a \pinm point $\pt$.
\begin{defn}
Let $I$ be a closed interval and $\partial I = \set{a,b}$. Fix a \pinm structure on $\partial I$ and \pinm
isomorphisms $\pt\cong a$ and $\pt\cong b$. A \term{relative \pinm structure} on $I$ is a \pinm structure on $I$
which restricts to the specified \pinm structure on $\partial I$. We consider two relative \pinm structures on $I$
equivalent if there's a \pinm diffeomorphism between them covering the identity and respecting this data, i.e.\ it
restricts to the identity on $\partial I$ and intertwines the \pinm diffeomorphisms with $\pt$.
\end{defn}
The \pinm diffeomorphisms $a\cong \pt\cong b$ define a \pinm structure on $I/\partial I$, and sending $I\mapsto
I/\partial I$ defines an isomorphism from the set of equivalence classes of relative \pinm structures on $I$ to the
set of diffeomorphism classes of \pinm structures on $S^1$; let $I_0$ be a relative \pinm structure on $I$ which
maps to $\Snb^1$ and $I_1$ be one which maps to $S_b^1$.
\begin{lem}
\label{combgroup}
Concatenation defines a group structure on the equivalence classes of relative \pinm intervals. This group is
isomorphic to $\Z/2$ and the generator is $I_0$.
\end{lem}
\begin{proof}
Because every 1-manifold $M$ can be stably framed relative to a fixed stable framing on the boundary, we may define
the \pinm structures on $I_0$ and $I_1$ as those induced by framings. Specifically, $I_0$ is induced by the trivial
framing (the restriction of the usual framing on $\R$ to $[0,1]$), and $I_1$ is induced by the nontrivial framing.
Two concatenated copies of this framing are equivalent to the trivial framing when the endpoints are fixed
(see~\cite[Remark 1.3.1]{DSPS13}) and concatenating with the trivial framing does not change the equivalence class
of framing on an interval, giving the claimed group structure.
\end{proof}

Let $M$ be a compact \pinm 1-manifold with a simplicial structure, and let $\Delta^i(M)$ denote its set of
$i$-simplices. For each $v\in\Delta^0(M)$, fix a \pinm isomorphism $v\cong\pt$. Since the groupoid of \pinm
structures on a point is equivalent to $\bullet/(\Z/2)$, an isomorphism with $\pt$ is a choice. For each
$e\in\Delta^1(M)$, the \pinm structure on $M$ defines a relative \pinm structure on $\overline e$. Thus $e\cong I_j$
for some $j\in\set{0,1}$; define $t(e)\coloneqq j$. From the function $t\colon\Delta^1(M)\to\Z/2$ one can recover
the \pinm structure on $M$ up to isomorphism. We will call $t$ the \term{combinatorial \pinm data} of $M$.
\begin{lem}
\label{mtospin}
Let $M$ be a spin circle with a simplicial structure and $m$ be the number of edges of $e$ with $t(e) = 1$.
\begin{itemize}
	\item If $m$ is odd, $M\cong\Sb^1$.
	\item If $m$ is even, $M\cong \Snb^1$.
\end{itemize}
\end{lem}
\begin{proof}
Fix a vertex $v\in M$. Using the group law from \cref{combgroup}, we can concatenate adjacent intervals for all
vertices except $v$, resulting in a simplicial structure on $M$ with a single vertex at $v$ and a single edge $e$
with $t(e) = m\bmod 2$. The result then follows from the definition of $t$.
\end{proof}

\subsection{Defining the Majorana chain}
\label{Majdefn}
Let $M$ be a compact \pinm 1-manifold with a simplicial structure. Associated to each vertex $v\in\Delta^0(M)$, we
associate a trivialized odd line $\C_v^{0\mid 1}$ and define the local state space $\cH_v\coloneqq\Lambda(\C_v)$.
The state space for the Majorana chain on $M$ is
\begin{equation}
	\cH\coloneqq \bigotimes_{v\in\Delta^0(M)} \cH_v.
\end{equation}
Let $F$ denote the space of functions $\Delta^0(M)\to\C$, regarded as a purely odd vector space. Then
$\cH\cong\Lambda^*(F)$, and hence $\cH$ is generated by the $\delta$-functions $\delta_v$ for $v\in\Delta^0(M)$,
where each $\delta_v$ is odd.
\begin{defn}
Let $v\in\Delta^0(M)$.
\begin{itemize}
	\item The \term{annihilation operator} associated to $v$, denoted $\iota_v$, is the interior product with
	$\delta_v$.
	\item The \term{creation operator} associated to $v$, denoted $\e_v$, is the exterior product with $\delta_v$.
	\item The \term{Majorana operators} associated to $v$ are
	\begin{align*}
		c_v &\coloneqq \e_v + \iota_v\\
		d_v &\coloneqq \e_v - \iota_v.
	\end{align*}
\end{itemize}
\end{defn}
\begin{rem}
The notation for the Majorana operators in~\cite{Kit01, FK11} corresponds to ours as follows: after ordering the
vertices $v_1,\dotsc,v_n$ on an interval in the direction defined by the orientation, their $c_{2j-1} = c_{v_j}$,
and their $c_{2j} = id_{v_j}$. In some papers, the Majorana chain is instead called the \term{Majorana wire} or
\term{Kitaev chain}.
\end{rem}
Let $\pi\colon M'\to M$ be the orientation double cover, and give $M'$ the simplicial structure which makes $\pi$ a
simplicial map. Then the $0$-skeleton of $M'$, $M_0'$, is a compact oriented $0$-manifold, hence comes with a
function $\fo\colon M_0'\to\set{\pm 1}$ sending a positively oriented point to $1$ and a negatively oriented point
to $-1$.

Let $n\coloneqq\abs{\Delta^0(M)}$.
\begin{lem}
The algebra generated by $c_v$ and $d_v$ is canonically isomorphic to $\Cl(\pi^{-1}(v), \fo)$ and isomorphic to
$\Cl_{1,1}$. The algebra generated by all Majorana operators is canonically isomorphic to $\Cl(M_0', \fo)$, and
noncanonically isomorphic to $\Cl_{n,n}$.
\end{lem}
\begin{proof}
If $V\in\Delta^0(M)$, let $v_+$, (resp.\ $v_-$) be the positively (resp.\ negatively) oriented preimage of $v$. We
define the maps $\ang{c_v,d_v}\to\Cl(\pi^{-1}(v), \fo)$ and $\ang{c_w,d_w\mid w\in\Delta^0(V)}\to\Cl(M_0', \fo)$ to
send $c_v\mapsto v_+$ and $d_v\mapsto v_-$. For this to define an isomorphism of algebras, one must check the
defining relations of the Clifford algebra: $c_v^2 = 1$, $d_v^2 = -1$, $[c_v, d_v] = -1$, and if $v\ne w$, $[c_v,
c_w] = [d_v, d_w] = [c_v, d_w] = -1$. These follow directly from the definition of the Majorana operators.

Since $\fo|_{\pi^{-1}(v)}$ sends $v_+\mapsto 1$ and $v_-\mapsto -1$, $\Cl(\pi^{-1}(v), \fo)\cong\Cl_{1,1}$ and
$\Cl(M_0', \fo)\cong\Cl_{n,n}$, the latter after choosing an ordering of the vertices of $v$.
\end{proof}

To define the Hamiltonian, we must orient $M$. This is a bit surprising, because the Majorana chain admits a
time-reversal symmetry and therefore ought to make sense on a \pinm manifold without using the fact that all
1-manifolds are orientable, but if we vary the orientation on $M$, we obtain a different Hamiltonian. We expect
that the eigenspaces for the Hamiltonian end up not depending on the choice of orientation.

The Hamiltonian for the Majorana chain is a sum of local terms for each edge. Fix an orientation on $M$, so that
each edge $e$ has an induced orientation; we write $\partial e = v -w$ to mean $\partial e = \set{v,w}$, and that,
in the induced orientation on the boundary, $v$ is the positively oriented boundary point and $w$ is the negatively
oriented one. For each $v\in\Delta^0(M)$, choose a \pinm isomorphism $v\cong\pt$, and let
$t\colon\Delta^1(M)\to\Z/2$ be the induced combinatorial \pinm data. Then, the Hamiltonian on $M$ is
\begin{equation}
\label{MCHam}
	H \coloneqq \frac{1}{2}\sum_{\substack{e\in\Delta^1(M)\\\partial e = v - w}} (-1)^{t(e)} c_v d_w.
\end{equation}
Time-reversal symmetry acts on $\cH$ as complex conjugation; since $c_v$ and $d_w$ are real, this commutes with the
Hamiltonian, so the Majorana chain admits a time-reversal symmetry squaring to $1$.
\begin{rem}
In physics, a Majorana fermion is a fermion which is its own antiparticle, meaning that its creation and
annihilation operators coincide. Because the Clifford relations imply $c_v^2 = 1$ and $(i\cdot d_v)^2 = 1$, these
operators can be interpreted as creating up to two Majorana fermions located at $v$. The Hamiltonian~\eqref{MCHam}
is expressing a relationship between Majorana fermions at adjacent vertices: if $\partial e = v - w$, then the
Hamiltonian specifies that low-energy states must have a relationship between the Majorana fermions corresponding
to $c_v$ and $i\cdot d_w$.

Because it would be interesting to observe a Majorana fermion, the Majorana chain has been studied
experimentally~\cite{MZFPBK, DYHLCX, DRMOHS, FHMJL, RLF}. To our knowledge, however, these experiments have not
considered the behavior of the Majorana chain under stacking or time-reversal symmetry.
\end{rem}
\subsection{The low-energy TQFT}
\label{Majlow}
We'd like to use \cref{bigansatz} to determine the deformation class of the low-energy theory $Z$ of the Majorana
chain, but it doesn't tell us everything. For example, neither \pinm structure on $\RP^2$ is bordant
to a disjoint union of mapping tori, so we won't be able to calculate $Z(\RP^2)$. Nonetheless, \cref{bigansatz}
tells us we can compute the state space of any closed 1-manifold and the partition functions of all \pinm tori
and Klein bottles. In particular, we'll find that $Z(\Snb^1)$ is an odd line, which is enough to imply that $Z$ is
one of the four generators of the $\Z/8$ of deformation classes of reflection positive 2D \pinm invertible field
theories. However, we can't be more specific: for all four generators, $Z(\Sb^1)$ is an even line, and $Z(T)$ and
$Z(K)$ are $1$ in the bounding \pinm structure and $-1$ in the nonbounding one.\footnote{Since $\Omega_2^\Spin$ and
$\Omega_2^{\Pin^+}$ are generated by mapping tori, this ambiguity does not appear for 2D spin and \pinp phases.
For general symmetry types, however, this is not the case, and additional work is needed to uniquely determine the
low-energy field theory of an SPT.}

Let $M$ be a spin circle with a simplicial structure, and let $t\colon\Delta^1(M)\to\Z/2$ be the combinatorial data
associated to it. Let $n\coloneqq \abs{\Delta^0(M)}$; then, the state space $\cH$ is a $\Z/2$-graded $\Cl(M_0',
\fo)$-module.
\begin{thm}[{\cite[\S5]{ABS}}]
\label{cliffrepn}
Up to isomorphism, $\ClHam$ has a single irreducible module $M$, which is $2^n$-dimensional. Up to even
isomorphism, $\ClHam$ has two irreducible supermodules, both isomorphic to $M$ after forgetting the
$\Z/2$-grading, and they are parity changes of each other.
\end{thm}
Since $\dim\cH = 2^n$, then $\cH$ is one of the two irreducible $\ClHam$-supermodules. The Hamiltonian acts on
$\cH$ as an element of $\ClHam$, since it's a sum of products of Clifford generators. Thus, to compute its
spectrum, it suffices to compute the action of $H\in\ClHam$
on any irreducible $\ClHam$-module $A$. To determine the parity of the space of ground states, we need to
know whether $\cH$ is graded isomorphic to $A$ or $\Pi A$, which we will do by fixing a grading operator
$\e\in\ClHam$ and comparing its action on $\cH$ and on $A$.

\begin{lem}
\label{SEndlem}
There is an isomorphism of superalgebras $\vp\colon \Cl_{1,1}\stackrel\cong\to\End(\C^{1\mid 1})$ sending
\begin{equation}
	v_+\mapsto \begin{pmatrix}0 & 1\\1 & 0\end{pmatrix}\qquad\qquad
	v_-\mapsto \begin{pmatrix*}[r]0 & 1\\-1 & 0\end{pmatrix*}.
\end{equation}
\end{lem}
\begin{proof}
One can directly verify that $\vp(v_\pm)$ are odd, $\vp(v_\pm)^2 = \pm I$, and $\vp(v_+)$ anticommutes with
$\vp(v_-)$.
\end{proof}
Thus, $g\coloneqq \vp(v_+v_-)$ is the grading operator on $\C^{1\mid 1}$.

Let $e\in\Delta^1(M)$ with $\partial e = v-w$. Since $\partial e$ is an oriented $0$-manifold, it comes with a
function $\fo_e\colon\partial e\to\set{\pm 1}$; the algebra generated by $c_v$ and $d_w$ is canonically isomorphic
to $\Cl(\partial e, \fo_e)$, which is isomorphic to $\Cl_{1,1}$. Let $\Cl(\partial e, \fo_e)$ act on a $\C^{1\mid 1}$ through the isomorphism from \cref{SEndlem}, and
call it $\C_e^{1\mid 1}$. Then there's a canonical isomorphism
\begin{equation}
	\ClHam \cong \bigotimes_{e\in\Delta^1(M)}\Cl(\partial e, \fo_e),
\end{equation}
so $\ClHam$ acts on
\begin{equation}
	A\coloneqq \bigotimes_{e\in\Delta^1(M)} \C_e^{1\mid 1},
\end{equation}
making $A$ into a graded $\ClHam$-module. Since $\dim_\C A = 2^n$, $A$ must be irreducible.

Let $m$ be the number of edges $e$ of $M$ with $t(e) = 1$.
\begin{prop}
\label{groundA}
Let $V\subset A$ denote the eigenspace for the smallest eigenvalue of $H$ acting on $A$. Then $V$ is
one-dimensional, and has parity $n-m\bmod 2$.
\end{prop}
\begin{proof}
Let $a\in A$ be a pure tensor of homogeneous elements of $\C_e^{1\mid 1}$. Then $a$ is homogeneous, so we let $\abs
a$ denote its degree, and for any $e\in\Delta^1(M)$, we let $\abs a_e$ be $1$ if the component of $a$ from
$\C_e^{1\mid 1}$ is odd, and $0$ otherwise.

If $\partial e = v-w$, then $c_vd_w$ acts by the grading operator $g_e$ on $\C_e^{1\mid 1}$, and therefore the
action of $H$ is
\begin{equation}
\label{HS1}
	H = \frac 12\sum_{e\in\Delta^1(M)} (-1)^{t(e)}\id\otimes\dotsb\otimes\id \otimes
	g_e\otimes\id\otimes\dotsb\otimes\id.
\end{equation}
It suffices to describe the action of $H$ on pure tensors of homogeneous elements, so let $a$ be such a tensor. If
$t(e) = 0$ for all edges $e$, then $H$ differs from $(n/2)\cdot\id$ on $a$ by subtracting $1$ for each odd
component of $a$. Therefore $H$ acts on $a$ as $(-1)^{n - 2\abs a}/2$, in which case the ground states are the
top-degree vectors, with eigenvalue $-n/2$.

More generally, if $e$ is an edge with $t(e) = 1$, it contributes $-g_e$ to $H$ instead of $g_e$. This change is
equivalent to multiplying by $(-1)^{2(2\abs a_e - 1)}$, so if $e_1,\dotsc,e_m$ are the edges with $t(e_i) = 1$,
then the action of the Hamiltonian on $a$ is
\begin{equation}
	H\cdot a = (-1)^{n - 2k + 2(2\abs a_{e_1} - 1) + \dotsb + 2(2\abs a_{e_m} - 1)}\frac{a}{2}.
\end{equation}
The $a$ which minimize the eigenvalue are those whose component in $\C_e^{1\mid 1}$ is odd if $t(e) = 0$ and odd if
$t(e) = 1$; these form a one-dimensional vector space with parity $n-m$.
\end{proof}
\begin{prop}
\label{paritymodule}
Let
\[\varepsilon\coloneqq\prod_{v\in\Delta^0(M)} d_vc_v\in\ClHam.\]
(The Clifford relations imply this doesn't depend on the order of the vertices in the product.) Then
\begin{itemize}
	\item on $\cH$, $\varepsilon$ acts on a homogeneous degree-$k$ element by multiplication by $(-1)^{n-k}$, and
	\item on $A$, $\varepsilon$ acts on a homogeneous degree-$k$ element by multiplication by $(-1)^{k-1}$.
\end{itemize}
\end{prop}
\begin{proof}
On $\cH$, $\varepsilon$ acts as
\begin{equation}
	(\varepsilon\cdot) = \prod_{v\in\Delta^0(M)} (\e_v - \iota_v)(\e_v + \iota_v) = \prod_v (\e_v\iota_v -
	\iota_v\e_v).
\end{equation}
It suffices to understand how this acts on pure wedges $\omega =
\lambda\delta_{v_{i_1}}\wedge\dotsb\wedge\delta_{v_{i_\ell}}$. On $\omega$, $\e_v\iota_v - \iota_v\e_v$ by the
identity if $v = v_{i_j}$ for some $j$, and by $-1$ otherwise. Therefore $\varepsilon\cdot\omega =
(-1)^{n-k}\omega$.

To compute the action of $\varepsilon$ on $A$, we rearrange it into a more convenient form. Choose a
$v_1\in\Delta^0(M)$, and let $v_2,\dotsc,v_n$ be the vertices encountered in order as one traverses the positively
oriented path around $M$ starting at $v_1$. Thus for each $i$, there's an edge $e_i$ with $\partial e_i =
v_{i+1\bmod n} - v_i$. Then,
\begin{align}
	\epsilon &= d_{v_1}c_{v_1}\dotsm d_{v_n}c_{v_n} = (-1)^n c_{v_1}d_{v_1}\dotsm c_{v_n}d_{v_n}.
\intertext{Since this string has $n$ letters, reversing it is a permutation of parity $(-1)^n$:}
	&= d_{v_n}c_{v_n}\dotsm d_{v_1}c_{v_1}.
\intertext{Finally, we commute $c_{v_1}$ past the remaining $2n-1$ operators:}
	&= -\underbracket{c_{v_1}d_{v_n}}_{g_n}{c_{v_n}}\dotsm \underbracket{c_3d_2}_{g_2} \underbracket{c_2d_1}_{g_1}.
\end{align}
Therefore $\varepsilon$ acts by $-1$ times the usual grading operator on $\C^{n\mid n}$ (i.e.\ the one which is
$-1$ on odd states).
\end{proof}
\begin{cor}
As graded $\ClHam$-modules, $\cH\cong\Pi^{n-1}A$, so the ground states of the Majorana chain on $M$ are
\begin{itemize}
	\item an even line if $m$ is odd (so $M\cong\Sb^1$), and
	\item an odd line if $m$ is even (so $M\cong\Snb^1$).
\end{itemize}
\end{cor}
\begin{proof}
In \cref{groundA}, we saw that the ground states of $H$ acting on $A$ have parity $n-m\bmod 2$, but by
\cref{paritymodule} the difference in the parities of $\cH$ and $A$ is $n-1\bmod 2$. Hence the ground state space
of $H$ acting on $\cH$ has parity $n-m-(n-1) = m-1$.
\end{proof}
\begin{cor}
Assuming \cref{bigansatz}, the low-energy TQFT $Z$ of the Majorana chain is a generator of the $\Z/8$ of
deformation classes of reflection positive \pinm invertible field theories. In particular, its deformation class is an odd multiple of
the class of the Arf-Brown theory.
\end{cor}
\begin{proof}
By a result of Schommer-Pries~\cite[Theorem 11.1]{SchInvertible}, we know $Z$ is invertible, since there is a \pinm
structure on $S^2$ and $Z(\Sb^1)$ and $Z(\Snb^1)$ are both invertible in $\sVect$. Since $Z_\AB$ generates the
$\Z/8$ of deformation classes of reflection positive 2D \pinm invertible TQFTs, $Z$ is deformation equivalent to
$(Z_{\AB})^{\otimes k}$ for some $k$, and is a generator iff $k$ is odd.

Because $Z_\AB(\Snb^1)$ is an odd line, then $(Z_\AB)^{\otimes k}(\Snb^1)$ has the same parity as $k$. Since
$Z(\Snb^1)$ is odd, then $k$ is odd.
\end{proof}
We can also study the Majorana chain on \pinm{} 1-manifolds with boundary, though again the Hamiltonian depends on
an orientation. Kitaev~\cite{Kit01} found that the space of ground states on an interval $I$ is two-dimensional;
from the low-energy perspective, this follows from the fact that for any choice of \pinm{} structure on $I$, $Z(I)$
is isomorphic to $\Cl_1$ as a $(\Cl_1,\Cl_1)$-bimodule. We can also see this directly from the lattice.

Suppose $n\coloneqq\abs{\Delta^0(I)}$. Orient $I$ and let $\partial I = v - w$. Then, $c_w$ and $d_v$ do not appear
in the Hamiltonian on $I$. Since each term in $H$ is $\pm 1/2$ times two Clifford generators not equal to $c_w$ or
$d_v$, both $c_w$ and $d_v$ commute with $H$, and therefore the algebra they generate, isomorphic to $\Cl_{1,1}$,
acts on all eigenspaces of $H$. In particular, if $V$ denotes the ground states of $H$, $V$ is a
$\Cl_{1,1}$-module, and by \cref{cliffrepn} is determined up to isomorphism by its dimension, which is even.

We can identify it with $\Cl_1$ in a manner similar to the proof of \cref{groundA}: define $A$ in the same manner
as above, except that we pretend there's an extra edge $e_\partial$ joining $v$ and $w$, so $A$ is a
$\Cl(I_0',\fo)$-module, where $I_0'$ is the $0$-skeleton of the orientation double cover $I'\to I$ and $\fo\colon
I_0'\to\set{\pm 1}$ is induced from the orientation as before. If $H_{S^1}$ denotes the Hamiltonian
from~\eqref{HS1} (for the circle), then our Hamiltonian is
\begin{equation}
	H = H_{S^1} - \id\otimes\dotsb\otimes\id\otimes g_{e_\partial},
\end{equation}
(where $t(e_\partial) \coloneqq 0$), whose action on a pure tensor of homogeneous elements $a\in A$ is
\begin{equation}
	H\cdot a = (-1)^{n - 2k + 2(2\abs a_{e_1} - 1) + \dotsb + 2(2\abs a_{e_m} - 1) + \abs
	a_{e_\partial}}\frac{a}{2}.
\end{equation}
Thus the ground state is two-dimensional, spanned by a pure tensor whose components are odd for all edges with
$t(e) = 0$ and even otherwise, and a pure tensor whose components are odd for all edges with $t(e) = 0$ except
$e_\partial$, and even otherwise. Since $\Cl_1$ is the unique two-dimensional irreducible
(ungraded) $\Cl_{1,1}$-representation up to isomorphism, the space of ground states on $I$ is isomorphic to either
$\Cl_1$ or $\Pi\Cl_1$. An argument with \cref{paritymodule} shows that we get the former. Finally, to match the
left $\Cl_{1,1}$-module description of the space of ground states with the $(\Cl_1,\Cl_1)$-bimodule description of
$Z(I)$, recall that a left $\Cl_{-1}$-action on a module $M$ is equivalent data to a right $\Cl_1$-action on $M$,
which implies the space of ground states on $I$ is $\Cl_1$ as a $(\Cl_1,\Cl_1)$-bimodule, in accordance with the
calculation using the low-energy TQFT.

\newcommand{\etalchar}[1]{$^{#1}$}

\end{document}